\documentclass[11pt,fleqn,letter]{article}

\usepackage{HLT24-Packages}
\usepackage{HLT24-Specific}

\usepackage{amsmath}

\title{\bf Hard CNF Instances for Ideal Proof Systems}

 \author{Tuomas Hakoniemi\thanks{Email: \url{tuomas.hakoniemi@helsinki.fi} This work has been supported by Helsinki Institute for Information Technology (HIIT).} \\ University of Helsinki
 \and 
 Nutan Limaye\thanks{Email: \url{nuli@itu.dk} This project is supported by funding from Independent Research Fund Denmark (grant agreement No. 10.46540/3103-00116B) and by the Basic Algorithms Research Copenhagen (BARC), which is funded by VILLUM Foundation Grant 54451.}\\ IT University of Copenhagen
 \and
 Iddo Tzameret\thanks{Email: \url{iddo.tzameret@gmail.com} This project has received funding from the European Research
Council (ERC) under the European Union’s Horizon 2020 research and innovation programme (grant agreement No 101002742,
EPRICOT project). It was also supported by the Engineering
and Physical Sciences Research Council (EPSRC) under grant
EP/Z534158/1, Integrated Approach to Computational Complexity:
Structure, Self-Reference and Lower Bounds}\\
 Imperial College London}

\usepackage{color}

\begin{document}

\maketitle

%___________________________________________________________
\begin{abstract}
Since the introduction of the Ideal Proof System (IPS) by Grochow and Pitassi (J.~ACM 2018), a substantial body of work has established size lower bounds for IPS and its fragments. In particular, Forbes, Shpilka, Tzameret, and Wigderson (Theory Comput.~2021) developed the main lower-bound frameworks for restricted IPS fragments, namely functional lower bounds and the hard multiples method, while Alekseev, Grigoriev, Hirsch, and Tzameret (SIAM J.~Comput.~2024) gave a general template for conditional lower bounds for full IPS.

Yet all these lower bounds apply only to purely algebraic formulas over a field, that is, non-Boolean formulas not directly expressible in propositional logic. Proving lower bounds for CNF formulas has therefore remained a central open problem in this line of work.

The current work resolves this question for IPS over read-once oblivious algebraic branching programs (roABPs) by proving lower bounds for refutations of CNF formulas in this system. 
Our approach is a rank-based feasible interpolation argument, following the method of Pudl\'ak and Sgall (Proof Complexity and Feasible Arithmetic 1996) 
for monotone span programs, in which decomposing a given roABP refutation along a variable partition yields a low-dimensional space of polynomials from which we construct a span-program interpolant. We extend their result from Nullstellensatz refutations measured by degree to Nullstellensatz refutations measured by roABP size (i.e., roABP-IPS$_\text{LIN}$).
\end{abstract}
%_______________________________________________________

%\newpage

%_______________________________________________________
\section{Introduction}
%_______________________________________________________

This work investigates lower bounds against algebraic proof systems in the framework of the Ideal Proof System ($\IPS$).
Proof complexity studies the size of proofs certifying membership in languages such as UNSAT, the set of unsatisfiable Boolean formulas. In this setting, a proof is an efficiently verifiable witness, and for UNSAT such a proof is usually called a refutation. A central goal of the area is to prove lower bounds for increasingly strong proof systems, with the ultimate aim of showing that no proof system has polynomial-size refutations for all unsatisfiable formulas. This line of work is often called \emph{Cook's programme}, following Cook's proposal that proof complexity lower bounds may shed light on basic complexity-theoretic problems such as $\P$ versus $\NP$. In particular, proving that no proof system efficiently refutes all unsatisfiable formulas would separate $\NP$ from $\coNP$, and hence also $\P$ from $\NP$.

Despite major progress, the main lower bound questions in proof complexity, especially those relating to strong enough proof systems such as textbook propositional logic remain open. Even proving lower bounds against propositional proofs operating with constant-depth formulas with counting modulo $p$ gates (i.e. \ACZ$[p]$-Frege) is open. In order to gain better understanding of the latter proof system, algebraic proof systems have been extensively studied beginning in the 1990s (cf.~\cite{BeameIKPP96}). Moreover, a recent important direction has been to study \emph{algebraic proof complexity}, with the hope to  bring methods from algebraic complexity theory proper into proof complexity (see \cite{PT16} for a survey). Here one ideally aims to shed light on proof systems that goes beyond \ACZ$[p]$-Frege, going up to Frege (i.e. textbook propositional proofs). 

The Ideal Proof System, introduced by Grochow and Pitassi \cite{GP18}, gives a particularly clean formulation of the  connection between algebraic circuits and propositional proofs. IPS can be viewed as a circuit-size version of Nullstellensatz refutations \cite{BeameIKPP96}, and hence as a proof system that directly incorporates algebraic circuit complexity into the proof complexity setting. In this way it gives a transparent reduction from circuit lower bounds to proof lower bounds as shown in \cite{GP18}: an IPS lower bound for a CNF formula implies that the permanent polynomial does not have polynomial-size algebraic circuits, namely, $\VP\neq\VNP$. Conversely, Santhanam and Tzameret~\cite{ST25} established a partial converse, giving further evidence that IPS forms a natural bridge between algebraic complexity and proof complexity. Namely, assuming $\VP \neq \VNP$, they constructed a family of CNF formulas admitting no polynomial-size IPS refutations. A qualification is that the unsatisfiability of this family remains open, although there is evidence supporting it.
The standard formulation of IPS used in this line of work is equivalent to circuit-size Nullstellensatz, and for restricted circuit classes one considers the corresponding fragments such as $\mathcal{C}$-$\lIPS$ (see below for more details).
 
Forbes, Shpilka, Tzameret and Wigderson~\cite{FSTW21} introduced two approaches for turning algebraic circuit lower bounds into lower bounds for $\IPS$: the \emph{functional lower bound} method and the \emph{lower bounds for multiples} method.
Other approaches include the \emph{meta-complexity} approach of~\cite{ST25}, which obtains a (conditional) $\IPS$ size lower bounds on a self-referential statement~(cf.~\cite{LST26}) and the \emph{noncommutative} approach of Li, Tzameret and Wang~\cite{LTW18} (building on \cite{Tza11-I&C}).

At the same time, the intended connection back to propositional logic has not yet been fully realised. All known unconditional lower bounds for IPS and its fragments apply to algebraic instances that are not direct encodings of propositional formulas, and in particular not of CNF formulas. Consequently, even a lower bound for IPS on such an algebraic instance does not by itself give a hard instance for propositional proof systems, such as fragments of Extended Frege: the hard instance is not a propositional formula, and hence is not an object in the usual language of propositional logic.

In fact, even moving beyond the \emph{single-axiom framework} has remained unclear. The lower bounds obtained via the functional lower-bound method are all based on a \emph{single} non-Boolean axiom, typically a variant of the subset-sum equation $\sum_i x_i=\beta$, where $\beta$ is chosen so that the equation has no $0$-$1$ solutions, for example $\beta=-1$. Thus, these results leave open the problem of proving IPS lower bounds for genuinely propositional instances. Andrews and Forbes~\cite{AF22} obtained lower bounds against constant-depth IPS, using the hard multiples method, for an algebraic instance with multiple axioms. However, their result is in the ``placeholder'' model of IPS refutations: the hard instance itself has no polynomial-size constant-depth algebraic representation, and is not
a translation of a propositional formula.
\medskip

The current work resolves this open problem for a natural fragment of IPS, studied in~\cite{FSTW21,HLT24,CGMS25,EGLT26}, by giving the first hard CNF instances in this fragment. We work with IPS refutations in which the certificate polynomials are computed by read-once oblivious algebraic branching programs (roABPs). More precisely, we prove lower bounds for $\roABP$-$\lIPS$ (defined below).

The argument goes back to classical results in propositional proof complexity about feasible interpolation, using an algebraic formulation by Pudl\'ak-Sgall~\cite{PS96-span-programs}.
To explain our result and setup we need to start by giving more details about algebraic proof systems and the feasible interpolation method in proof complexity.

\para{Algebraic proof systems.}

Algebraic proof systems certify that a given set of multivariate polynomials over a field has no common Boolean solutions.
Some of the basic proof systems in this line are the Polynomial Calculus (PC) \cite{CEI96} and its `static' variant, Nullstellensatz \cite{BeameIKPP96}.
In PC, proofs proceed by algebraic manipulation, adding and multiplying polynomials, until deriving the contradiction $1 = 0$.
Contrastingly, in Nullstellensatz,  a proof of the unsatisfiability of a set of axioms, written as polynomial equations $\{f_i(\vx)=0\}$ over a field, is a \emph{single} formal polynomial identity expressing $1$ as a combination of the axioms, that is:
\begin{equation}
\label{eq:NS-first}
\sum_i g_i(\vx)\cdot f_i(\vx) = 1\,,
\end{equation}
for some polynomials $\{g_i(\vx)\}$.
These systems measure proof size by sparsity, defined as the total number of monomials involved, which makes them comparatively weak.
An alternative way to measure proof size is by  algebraic circuit size. This was suggested initially by Pitassi \cite{Pit97,Pit98}, and further investigated in the work of Grigoriev and Hirsch \cite{GH03} and subsequently Raz and Tzameret \cite{RT06,RT07}, eventually leading to the Ideal Proof System \cite{GP18} described in what follows.

\para{Ideal Proof System.}
The Ideal Proof System ($\IPS$ for short; \Cref{def:IPS}), introduced by Grochow and Pitassi \cite{GP18}, loosely speaking, is the Nullstellensatz proof system where the polynomials $g_i(\bar{x})$ in \Cref{eq:NS-first} are represented by algebraic circuits. Forbes, Shpilka, Tzameret and Wigderson~\cite{FSTW21} showed that these two representations are formally equivalent.
%showed that $\IPS$ is equivalent to Nullstellensatz %in which the polynomials $g_i$ in  \Cref{eq:NS-first} %are written as algebraic circuits. 
In other words, an $\IPS$ refutation of the set of axioms $\left\{f_i(\bar{x})=0\right\}_i$ can be defined similarly to  \Cref{eq:NS-first}:
\begin{equation}
\label{eq:intro:IPS}
\sum_i g_i(\bar{x}) \cdot f_i(\bar{x})+\sum_j h_j(\bar{x}) \cdot\left(x_j^2-x_j\right)=1,
\end{equation}
for some polynomials $\left\{g_i(\bar{x})\right\}_i$, where we think of the polynomials $g_i, h_j$ written as algebraic circuits (instead of e.g., counting the number of monomials they contain); and where $x_j^2-x_j$ are the set of \emph{Boolean axioms}, forcing every solution to the equations $\{x_j^2-x_j=0 \;:\; j\in [n]\}$ to 0-1 values.
Thus, the size of the $\IPS$ refutation in \Cref{eq:intro:IPS} is $\sum_i \operatorname{size}\left(g_i(\bar{x})\right)+\sum_j \operatorname{size}\left(h_j(\bar{x})\right)$, where $\operatorname{size}(g)$ stands for the minimal size of an algebraic circuit computing the polynomial $g$.

Note that IPS is \emph{not} necessarily a Cook-Reckhow propositional proof system, in the sense that there is no known deterministic polynomial-time algorithm to check whether a refutation is a correct refutation. There reason is that to verify \Cref{eq:intro:IPS} we need to preform polynomial identity testing (PIT for short), which is a problem in $\coRP\subseteq\BPP$, not known to be in $\P$. 

It is natural to consider IPS refutations where 
the polynomials $g_i(\vx)$ and $h_j(\vx)$ in \Cref{eq:intro:IPS} are written as algebraic circuits from a prescribed circuit class $\mathcal{C}$, for example constant-depth algebraic circuits. However, when considering $\mathcal{C}$ weaker than general algebraic circuits, one has to be a bit careful with the definition of $\IPS$. For technical reasons the formalization in \Cref{eq:intro:IPS} does not capture the precise definition of $\IPS$ restricted to the relevant circuit class, rather the fragment which is denoted by $\mathcal{C}$-$\lIPS$ (``LIN" here stands 
for the linearity of the axioms $f_i$ and the Boolean axioms; that is, they appear as polynomials $f_i^d$ with power $d=1$). In this work, we focus on $\mathcal{C}$-$\lIPS$. Namely, refutations in the system $\mathcal{C}$-$\lIPS$ are defined as in    \Cref{eq:intro:IPS} where the polynomials $g_i(\vx), h_j(\vx)$ are written as circuits in the circuit class $\mathcal{C}$. In particular we shall consider $\mathcal{C}$ to be the class of roABPs described in what follows.

\para{Read-once oblivious algebraic branching programs.}
A read-once oblivious algebraic branching program (denoted, roABP) is a layered directed graph with two special vertices, the source in the first layer and the sink in the last layer, and $n - 1$ layers in between, and where edges are directed from source to sink. All nodes in layer $i$ are connected with edges only to nodes in layer $i+1$. The edges between layer $i$ and layer $i + 1$ are labeled with univariate polynomials in the variable $x_{i}$ (i.e., in each layer all polynomials have the same (single) variable). Each source-to-node  path computes the polynomial that is the product of all the univariate polynomials on the path, and the polynomial computed at each node is the sum of all paths incoming into the node. The program itself computes the polynomial computed at the sink node. 
The \emph{width} of an roABP is the maximum number of nodes in any layer.

Note that the order in which variables are read in a roABP is fixed. In the above description we made the order to be from $x_1$ to $x_n$, namely,  by increasing index order, but equivalently we can fix any other linear (i.e., total) order.

The proof system roABP-\lIPS\ defines refutations as in \Cref{eq:intro:IPS} in which $g_i(\vx)$ and $h_j(\vx)$ are written as roABPs. This is an interesting  system, because unlike IPS, roABP-\lIPS\ is known to be a Cook-Reckhow proof system: deterministic polynomial-time PIT algorithms for roABPs are known (cf.~Raz and Shpilka~\cite{RS04}). Moreover, roABP-\lIPS\ is a test case for lower bounds against IPS fragments, because of the relatively simple nature of roABPs, and the fact that lower bounds against this circuit class are well understood. 

\para{Feasible interpolation.}

Let $\varphi_0(\vx,\vz)\wedge\varphi_1(\vy,\vz)$ be an unsatisfiable propositional formula in pairwise disjoint sequences of variables $\vx$, $\vy$ and $\vz$. 
It follows that for every given total assignment to $\vz$, either $\varphi_0(\vx,\vz)$ or $\varphi_1(\vy,\vz)$ is unsatisfiable (or both are); this is because, a single $\vz$-assignment for which both $\varphi_0(\vx,\vz)$ and $\varphi_1(\vy,\vz)$ are satisfiable, could be merged into a satisfying assignment to $\varphi_0(\vx,\vz)\land\varphi_1(\vy,\vz)$.

The basic interpolation property for propositional logic, analogous to Craig interpolation in first-order logic, guarantees the existence of a function $I$ defined on the $\vz$-variables such that:
\[
I(\valpha) =
\begin{cases}
    0, \text{ if }\varphi_0(\vx,\valpha)\text{ is satisfiable;}\\
    1, \text{ if }\varphi_1(\vy,\valpha)\text{ is satisfiable.}
\end{cases}
\]

Such a function $I$ is called \emph{an interpolant} for the conjunction $\varphi_0(\vx,\vz)\wedge\varphi_1(\vy,\vz)$. The function is well-defined, i.e. single-valued, by the assumption that the conjunction is unsatisfiable as explained above.
In a similar manner, we can consider feasible interpolation for sets of polynomial equations instead of propositional formula(s) (see bottom of \Cref{sec:perlim:feasible-interpolation}).

In the framework of \emph{feasible interpolation} we are interested in the computational complexity of the interpolants in terms of the sizes of the refutations of the propositional formulas. As such the framework is a method to translate computational lower bounds to proof size lower bounds.
See \Cref{sec:conn-to-previous-work} for more on the history of feasible interpolation.

\para{Span programs.}
A span program \cite{KW93} consists of a labeled matrix $(M,\sigma)$ together with a target vector $t$, where $\sigma$ is a labeling of the rows, i.e., a mapping from the rows of $M$ to the literals in the variables $\vz$. We also allow labeling a row with the constant $1$. A span program is monotone, whenever there is no row labeled with a negative literal. The \emph{size} of the span program is the number of rows (irrespective of the number of columns).

Equivalently, one may define span programs without choosing coordinates.
In this formulation, a span program consists of a finite set $U$ of pairs
$(\ell,v)$, where $\ell$ is a $\vz$-literal or the constant $1$, and
$v$ is a vector in some vector space $V$, together with a target vector
$t\in V$. This is equivalent to the matrix formulation: a matrix
representation is obtained by choosing a basis of $V$ and writing the
vectors $v$ as coordinate rows, while conversely the rows of a labeled
matrix are simply labeled vectors in the ambient coordinate space. We use
this coordinate-free formulation below because our vectors will naturally
be polynomials, and no particular choice of monomial basis will be relevant.

A span program as above \emph{computes} a Boolean function as follows: on a 0-1 assignment $a$ to the variables $\vz$, let $U(a)$ be the set of vectors $v$ so that there is some literal $\ell$ with $(\ell,v)\in U$ and $\ell(a) = 1$. In other words, $U(a)$ consists of all vectors picked by the assignment $\vz$, namely those vectors whose $\vz$-label gets 1 under $a$. The span program outputs $1$ if the target vector $t$ is in the span of the vectors in $U(a)$ and $0$ otherwise.

%%%%%%%%%%%%%%%%%%%%%%%%%%%%%%%%
\subsection{Our Results}
%%%%%%%%%%%%%%%%%%%%%%%%%%%%%%%%

Our main result  is an exponential lower bound against roABP-$\lIPS$ for a family of CNF formulas, for any variable order and over any field. Namely, given the CNF formula, no roABP-\lIPS\ certificate of any variable order and sub-exponential size exists.

\begin{theorem}[Main lower bound (informal); see \Cref{thm:main-final-lower-bound}]
\label{intro:thm:main-final-lower-bound}
There is a family of CNF formulas $\Psi_n$ in $\poly(n)$ variables and $\poly(n)$ clauses that requires \textup{roABP}-$\lIPS$ refutations of size $2^{n^{\Omega(1)}}$ in any variable ordering, and over any field.
\end{theorem}

This resolves a problem left open by previous works on \IPS\ lower bounds~\cite{FSTW21,HLT24,For24,And25,EGLT26}; see also the discussion in Andrews~\cite[Sec.~1.3]{And25}. Namely, whether one can prove an \IPS\ lower bound for an instance that is not a single polynomial equation and, more importantly, for a genuinely Boolean instance: a direct arithmetization of a propositional formula, specifically a CNF formula.

The significance of this result is twofold. First, CNF formulas are the standard benchmark instances in proof complexity. Second, as discussed above, they provide a direct connection to Frege-style propositional proof systems: a lower bound for $\mathcal C$-IPS on CNFs immediately implies the same lower bound for any propositional proof system simulated by $\mathcal C$-IPS. This implication is unavailable for purely algebraic instances, such as subset-sum equations of the form $\sum_i a_i x_i-\beta$ and their variants.

For example, a lower bound against constant-depth IPS over $\F_p$ for a CNF would imply a lower bound for $\ACZ[p]$-Frege, a longstanding open problem. Although recent work of Elbaz \textit{et al.}~\cite{EGLT26} shows that even non-Boolean, namely purely algebraic, instances that are hard for constant-depth IPS over finite fields imply a CNF lower bound via a general translation lemma, it remains important to develop \emph{direct} proof-size lower-bound methods for CNF formulas, in the hope that such methods may eventually yield lower bounds against constant-depth IPS refutations.

\bigskip

The family of hard CNF formulas $\Psi$ is based on a lifting of the instance $\GEN_n$ introduced in Raz and McKensize~\cite{RazMcKenzie99} (\Cref{def:GEN}) and used throughout circuit and proof complexity (see in particular Pitassi-Robere~\cite{PitassiRobere2018}, Chan-Potechin~\cite{ChanPotechin2014} and  Robere, Pitassi,  Rossman, and  Cook \cite{robere_exponential_2016}).  

\smallskip
The  proof of \Cref{intro:thm:main-final-lower-bound} has three main steps:
\begin{enumerate}
\item 
Interpolation;
 
 \item Hard monotone instance;  
 
 \item Lifting.
\end{enumerate}

We describe these three steps in what follows.

\para{Interpolation.} We establish a general feasible interpolation result as follows. Let \(\Psi=\phi_0(\vx,\vz)\wedge \phi_1(\vy,\vz)\) be  any split formula (with $P_0$ being a set of sparse polynomials; e.g., when $\Psi$ is a $k$-CNF, for a small $k$). From a short roABP-$\lIPS$ refutation with a fixed variable order, refuting $\Psi$, we show how to extract a small span program over the input $\vz$ variables  separating inputs for which $\phi_0(\vx,\vz)$ is satisfiable from those inputs for which $\phi_1(\vy,\vz)$ is satisfiable.
When the $\vz$ variables are all positive in $\phi_0$ or all negative in $\phi_1$ we extract a \emph{monotone} span program, hence we call it ``monotone feasible interpolation''.

\begin{theorem}[Feasible interpolation for roABP-$\lIPS$; see \Cref{thm:feasible-interpolation}]
\label{intro:thm:feasible-interpolation}
    Let $P_0(\vx,\vz)$ and $P_1(\vy,\vz)$ be sets of polynomial equalities, where $\vx,\vy$ and $\vz$ are pairwise disjoint variables and $P_0$ is a set of sparse polynomials of polynomial degree (e.g.,~a translation of $k$-CNF, for a constant  $k$). Suppose that $P_0\cup P_1$ has an \textup{roABP}-$\lIPS$ refutation of width $w$ in a variable ordering, where $\vx$ variables precede all other variables. Then there is a span program of size $\poly(w\cdot |P_0(\vx,\vz)|)$ that computes an interpolant for $P_0$ and $P_1$. Moreover, if all the polynomials in $P_0$ are monotone in $\vz$ (see \Cref{eq:monotone}), then the span program is monotone. 
\end{theorem}

\noindent \textit{Proof idea:} The proof  has one central idea. Start with an roABP-$\lIPS$ refutation
\[
\sum_{p\in P_0} a_p p+\sum_{q\in P_1} b_q q=1
\]
in a variable order where all \(\vx\)-variables come first. First, put \(P_0\) into \(\vz\)-normal form, so each polynomial is either \(p_0(\vx)\) or \(p_0(\vx)+z_jp_1(\vx)\), for some polynomials $p_0(\vx), p_1(\vx)$ in the displayed variables \vx\ only. Then use the roABP structure: since each \(a_p\) is computed by a width-\(w\) roABP and all \(\vx\)-variables  appear first, cutting the roABP right after the last \(\vx\)-layer gives
\[
a_p=\sum_{i\in[w]} a_{p,i}(\vx)\,r_i(\vy,\vz).
\]
The span program is built from the \(\vx\)-parts \(a_{p,i}\): for each \(p_0(\vx)+z_jp_1(\vx)\), add vectors labeled by \(z_j\) and \(1-z_j\), and for each \(p_0(\vx)\), add a constant-labeled vector. The target vector is \(1\).

Why does this work? If the span program outputs \(1\) on an assignment \(\valpha\) to \(\vz\), then the selected vectors already express \(1\) as a linear combination of the specialized polynomials from \(P_0(\vx,\valpha)\), so \(P_0(\vx,\valpha)\) is unsatisfiable. Conversely, if it outputs \(0\) but \(P_1(\vy,\valpha)\) were satisfiable, plugging a satisfying \(\vy\)-assignment into the original refutation would kill the \(P_1\)-part and leave a refutation of \(P_0(\vx,\valpha)\); by the decomposition above, this refutation is again a linear combination of the selected vectors, which contradicts the fact that the span program outputs 0. Thus the span program computes an interpolant. If \(P_0\) is monotone in \(\vz\), the construction is monotone as well. 

\para{Hard monotone instance.}
To make use of the above feasible monotone interpolation in order to establish an actual proof size lower bound we need to find a monotone function that we know is hard for monotone span programs. Moreover, we must be able to find a suitable polynomial-size split formula $\Psi$, of which its interpolant is the hard monotone function. The idea is to use the function $\GEN_n$ that was shown hard for monotone span programs by Robere, Pitassi, Rossman and Cook~\cite{robere_exponential_2016,PitassiRobere2018}. Unlike \cite{robere_exponential_2016,PitassiRobere2018},
 we do not need to insist on a function that is computable by  small monotone circuits, rather  a general (non-monotone) circuit suffices. Hence, we can use the general function $\GEN_n$ instead of the pyramid-based one in the previous works (see \Cref{sec:fixed-order}). 

To form the split formula $\Psi=\phi_0(\vx,\vz)\wedge \phi_1(\vy,\vz)$ whose interpolant is the hard instance for monotone span programs we define it as follows. 
The formula  \(\phi_0(\vx,\vz)\) describes the negative instances of \(\GEN_n\) and \(\phi_1(\vy,\vz)\) describes its positive instances. Also, because $\GEN_n$ is monotone, the associated formulas can be constructed  to be monotone in a sense that makes monotone feasible interpolation applicable. Hence the interpolant extracted from a small refutation 
is a a small monotone span program separating no-instances from yes-instances of \(\GEN_n\), shown hard in \cite{robere_exponential_2016,PitassiRobere2018}.

\para{Lifting.} In the previous step we obtained a split formula that is hard for a specific order of variables (when the \vx-variables come before the other variables). This does not rule out the possibility of a small refutation for a different variable ordering. To rule this out, we apply a ``lifting'' process on the split formula $\Psi$ that, roughly, embeds all possible variable orders (see~\cite{FSTW21} and subsequent work on lifting in IPS; and cf.~\cite{RGR22} for general lifting in proof complexity). 
We explain it informally in what follows. 

Our aim is to remove the dependence on a fixed variable order. In the previous step the hard CNF
$\phi=\phi_0\wedge\phi_1
$
is only shown to be hard for roABP-$\lIPS$ refutations when the variables are read in one specific order, say $x_1<\cdots<x_N$. Lifting constructs a larger CNF formula 
\[
\Psi(\vv,\vu)=\Phi_1(\vv,\vu)\wedge\Phi_2(\vu)
\]
such that, for every variable order $<$ on the variables of $\Psi$, there is a substitution $\rho_{<}$ with the following effect: $\rho_{<}$ assigns Boolean values to the auxiliary variables $\vu$, relabels the $\vv$-variables by $x_1,\dots,x_N$ according to their order under $<$, and after this substitution the lifted clause part $\Phi_1$ becomes exactly the original hard CNF formula $\phi$, while the ``consistency part'' $\Phi_2$ (see below) is satisfied. Thus $\phi$ appears as a restricted instance of $\Psi$. Consequently, any roABP-$\lIPS$ refutation of $\Psi$ in order $<$ yields, under $\rho_{<}$, a refutation of $\phi$ in its hard order.

The formula $\Phi_1(\vv,\vu)$ is the actual lift of $\phi$. It uses new variables $\vu$ and $\vv$. The $\vv$-variables are the new main variables, whose order inside the refutation may be arbitrary. The $\vu$-variables are selector variables. For each clause $C=(\ell_1\vee \ell_2\vee \ell_3)$ of $\phi$, and for each triple $(i,j,k)$ of distinct indices, the variable $u_{C,i,j,k}$ says that the first, second, and third literals of $C$ are realised on $v_i,v_j,v_k$, respectively, with the same signs as in $C$ (``realised'' in the sense that there will be a relabelling of $\vv$-variables by $\vx$-variables to recover the original instance). This is enforced by the implication
\[
u_{C,i,j,k}\to \bigl(v_i^{\mathrm{sg}(\ell_1)}\vee v_j^{\mathrm{sg}(\ell_2)}\vee v_k^{\mathrm{sg}(\ell_3)}\bigr),
\]
together with clauses saying that for each original clause $C$, exactly one such triple is chosen. Hence $\Phi_1$ embeds each clause of $\phi$ locally into the $\vv$-variables.

However, these local choices could be inconsistent across different clauses. The same original variable might correspond to different $\vv$-variables in different places. The role of $\Phi_2$ is to prevent this. It is an auxiliary consistency gadget, not part of the original split $\phi_0\wedge\phi_1$. It checks that the selected $\vu$-variables are globally consistent, meaning that local correspondences for each original clause give rise to a global one-to-one correspondence between the original variables $\vx$ and the new variables $\vv$.
% whenever two clauses of $\phi$ share a variable, their chosen triples send that variable to the same $v_i$. 
% This condition is encoded by a polynomial-size CNF using fresh $\vw$-variables. 
Thus $\Phi_2$ is satisfiable exactly when the local choices in $\Phi_1$ come from one single global relabelling.

Now fix an arbitrary variable order $<$ on $\Psi$, and let the induced order on the $\vv$-variables be
\[
v_{i_1}<v_{i_2}<\cdots<v_{i_N}.
\]
This order determines the restriction $\rho_{<}$. First, relabel $v_{i_k}\mapsto x_k$ for each $k\in[N]$. Next, for each clause $C$ of $\phi$, set exactly one selector variable $u_{C,i,j,k}$ to $1$, namely the one corresponding to the positions of the three variables of $C$ in this $\vv$-order, and set all other selector variables for $C$ to $0$. Since this choice comes from one global relabelling, it is automatically globally consistent, and therefore $\Phi_2$ can be satisfied by a suitable assignment to the $\vu$-variables. After this restriction, the auxiliary gadget disappears and the lifted formula reduces to the original hard CNF $\phi$, with the surviving variables now ordered as $x_1<\cdots<x_N$.

Therefore the restriction depends on the variable order $<$, but not on any further feature of the proof. Once the order of the $\vv$-variables is fixed, the restriction is fixed as well. Hence a small roABP-$\lIPS$ refutation of $\Psi$ in any order would yield a small roABP-$\lIPS$ refutation of $\phi$ in its hard order, contradicting the lower bound for a fixed order we obtained in the previous step. Hence, $\Psi$ is hard for roABP-$\lIPS$ in every variable order. 

% ========================================
\subsection{Connection to Previous Work}\label{sec:conn-to-previous-work}
% ========================================

Feasible interpolation in proof complexity arose from the interaction between bounded arithmetic and propositional proof systems. Early background appears in Kraj\'\i\v{c}ek--Pudl\'ak~\cite{KP89} and in Razborov's work from the mid-1990s~\cite{Razb95-Bounded,Razb95-Unprovability}. It was developed as a proof-complexity lower-bound method by Bonet, Pitassi, and Raz~\cite{BPR97}, and formulated explicitly in Kraj\'\i\v{c}ek~\cite{Kra97-Interpolation}. These works formulated the method via communication complexity. Pudl\'ak~\cite{Pud97} then gave a structural interpolation theorem that extracts computation from proofs without going through communication complexity. Our feasible interpolation theorem is most closely related to the monotone interpolation theorem for algebraic proof systems due to Pudl\'ak and Sgall~\cite{PS96-span-programs}, that extracts a small (monotone) span program from a low degree algebraic proof.  

Feasible interpolation continued to play a major role in proof complexity throughout the years. In particular, lower bounds on OBDD-based proof systems by Kraj\'{i}\v{c}ek~\cite{Kra2007} (cf.~Segerlind \cite{Seg2007,Seg08-OBDD-vs-resolution} for the tree-like case), used variable ordering arguments close to our constructions. OBDD-based propositional proof systems were introduced by Atserias, Kolaitis, and Vardi~\cite{AKV04} and further investigated in several later works (cf.~\cite{BIKRS21,Its17-STACS,IR22}).

The ability to apply monotone feasible interpolation for monotone span programs stems from the results of Robere, Pitassi, Rossman and Cook, and Pitassi-Robere~\cite{robere_exponential_2016,PitassiRobere2018}.  

As for \IPS\ refutations and specifically roABP-\lIPS, this system was considered first in Forbes, Shpilka, Tzameret and Wigderson~\cite{FSTW21}. They proved an exponential size lower bound for this system alas for a single formula. All lower bounds against IPS and fragments are for either a single algebraic instance \cite{FSTW21,GHT22,HLT24,BLRS25,EGLT26,CGMS25,BHLS26}; or for non-single axiom as in Andrews-Forbes~\cite{AF22}, but there the instance is algebraic and the lower bound is in the ``placeholder model", meaning that the hard instance against $\mathcal C$-IPS cannot be written by itself with a (small) circuit in the class $\mathcal C$.  

% ========================================
\subsection{Concluding Discussion and Open Problems}
% ========================================

The lower bound method developed here shows that, for roABP-$\lIPS$, feasible interpolation is not far from rank-based approaches as used in all previous (unconditional) IPS lower bounds (namely, by reducing IPS refutation size lower bound  to an algebraic circuit lower bound, where the latter is established using a rank-argument). The key step in the interpolation theorem is to cut each coefficient roABP at the interface between the \(\vx\)-variables and the remaining variables. This yields a decomposition through a space of dimension at most the width of the roABP, and this low-dimensional space is then converted into a span-program interpolant. Thus, at a higher level, the argument has the same structure as a rank-argument used throughout algebraic circuit complexity and previous IPS lower bounds: a small refutation gives rise to a low-dimensional object, while hardness of the interpolant implies that every such object must be large. In our case, the first ingredient is the roABP decomposition, and the second is the monotone span-program lower bound for \(\GEN_n\).

This also gives a more concrete explanation of the remark in \cite{FSTW21} that the functional lower-bound method is reminiscent of feasible interpolation. At least for roABP-based IPS, the current work decreases the gap between these two methods. 

It is also instructive to compare the present result with lower bounds for OBDD-based proof systems. The two settings are related in that both are sensitive to variable order, and both admit lower-bound arguments that exploit this feature. However, the systems themselves are  orthogonal. The proof system studied here is algebraic and static, and it works with coefficients represented by roABPs, which are strictly more expressive than OBDDs. By contrast, OBDD proof systems are propositional, dynamic, and semantic. Therefore, the present lower bound should not be understood as a direct analogue of OBDD lower bounds, but rather as an algebraic proof-complexity result that shares with them only the order-sensitive aspect of the argument.

Finally, the lower bound does not necessarily reflect a general weakness of roABP-$\lIPS$. The system also admits nontrivial upper bounds as we show in this work. In particular, as we show in  \Cref{sec:sim-and-upper-bound}, Tseitin formulas have polynomial-width roABP-$\lIPS$ refutations in every variable order: one can start from Grigoriev's polynomial-sparsity Nullstellensatz refutation in the \(\{\pm 1\}\) basis and then apply a linear change of basis, which preserves small roABP width. 
And further roABP-$\lIPS$ simulates tree-like PC proof length (number of lines) and efficiently refutes the functional pigeonhole principle formulas.
\medskip 

Regarding \emph{questions left open}, the immediate problem is to see if feasible interpolation could obtain CNF formulas hard against other fragments of IPS for which lower bounds are already known for single-axiomed algebraic instances. These fragments include multilinear formula IPS \cite{FSTW21,HLT24} and constant-depth IPS refutations with low individual degrees as in \cite{GHT22,HLT24,EGLT26}. 

On the other hand, strong enough proof systems, such as constant-depth IPS with no bounds on individual degrees simulate $\ACZ[p]$-Frege (when the IPS system works over the field $\F_p$; see \cite{GP18}) and plausibly \TCZ-Frege (when the IPS system works over characteristic 0 fields). The latter system \TCZ-Frege is known not to have feasible interpolation property based on the hardness of 
factoring \cite{BPR00} (i.e., non-feasible interpolation follows assuming factoring Blum integers is not in $\P/\poly$), and similarly the former system \ACZ$[p]$-Frege does not have feasible interpolation unless the Diffie-Hellman function can be computed by subexponential-size circuits, i.e.,  $2^{n^\eps}$, for arbitrarily small $\eps>0$ \cite{BDGMP99}. Hence, we cannot realistically hope for feasible interpolation to hold for (unrestricted individual degree) constant-depth IPS refutations.

\subsection{Follow-ups}
After our work was presented in a talk, Dmitry Sokolov \cite{Sok26} pointed out  a way to lift general Nullstellensatz degree lower bounds to roABP-\lIPS\ refutation size lower bounds for a fixed variable order.  
This is based on the algebraic tiling approach of Pitassi and Robere~\cite{PitassiRobere2018} (cf.~\cite{robere_exponential_2016}). In this sense, one  obtains CNF formulas hard against roABP-\lIPS\ in a fixed variable order by lifting any CNF formula that is hard for Nullstellensatz degree. Using the method developed in \Cref{sec:any-order-roABP}, one can then extend the resulting lower bound to arbitrary variable orders.

% ========================================
\section{Preliminaries}
% ========================================

\subsection{Notation}

We use lower case over-lined Latin letters $\vx,\vy,\vz$, etc.~to denote sequences of (Boolean) variables. Single variables always have a subscript attached, e.g.~$x_i$ or $y_{j_0}$. We will use lower case over-lined Greek letters $\valpha,\vbeta,\vgamma$ etc.~to denote Boolean assignments, whose domain will be clear from the context. 

\subsection{Feasible interpolation}\label{sec:perlim:feasible-interpolation}

Let $\varphi_0(\vx,\vz)\wedge\varphi_1(\vy,\vz)$ be an unsatisfiable propositional formula in pairwise disjoint sequences of variables $\vx$, $\vy$ and $\vz$. We call such formulas \emph{split formulas.}
It follows that for every given total assignment to $\vz$, either $\varphi_0(\vx,\vz)$ or $\varphi_1(\vy,\vz)$ is unsatisfiable (or both are); this is because, a single $\vz$-assignment for which both $\varphi_0(\vx,\vz)$ and $\varphi_1(\vy,\vz)$ are satisfiable, could be merged into a satisfying assignment to $\varphi_0(\vx,\vz)\land\varphi_1(\vy,\vz)$.

The basic interpolation property for propositional logic, analogous to Craig interpolation in first-order logic, guarantees the existence of a function $I$ defined on the $\vz$-variables such that: 
\[
I(\valpha) =
\begin{cases}
    0, \text{ if }\varphi_0(\vx,\valpha)\text{ is satisfiable;}\\
    1, \text{ if }\varphi_1(\vy,\valpha)\text{ is satisfiable.}
\end{cases}
\]
% For the inputs that don't satisfy either condition, map them arbitrarily to $0$ or $1$. 
Note that this is identical to the following two conditions:
\begin{align}\label{eq:contrapositive-interpolation}
&\text{ If } I(\valpha)=0\text{, then } \varphi_1(\vy,\valpha) \text{ is unsatisfiable; and}\\\notag
&\text{ If } I(\valpha)=1\text{, then } \varphi_0(\vy,\valpha) \text{ is unsatisfiable}.
\end{align}
Such a function $I$ is called \emph{an interpolant} for the conjunction $\varphi_0(\vx,\vz)\wedge\varphi_1(\vy,\vz)$. The function is well-defined, i.e. single-valued, by the assumption that the conjunction is unsatisfiable as explained above.

In the framework of \emph{feasible interpolation} we are interested in the computational complexity of the interpolants in terms of the sizes of the refutations of the propositional formulas. As such the framework is a method to translate computational lower bounds to proof size lower bounds. We say that a proof system $\mathcal{P}$ admits a feasible interpolation property with respect to a circuit model $\mathcal{C}$, if for any $\mathcal{P}$-refutation of size $s$ of a split formula $\varphi_0(\vx,\vz)\wedge\varphi_1(\vy,\vz)$ there is a $\mathcal{C}$-circuit of size $\poly(s)$ that computes an interpolant for $\varphi_0(\vx,\vz)\wedge\varphi_1(\vy,\vz)$. In the case that there exists monotone function interpolating $\varphi_0(\vx,\vz)\wedge\varphi_1(\vy,\vz)$ we say that the proof system $\mathcal{P}$ admits a \emph{monotone feasible interpolation} with respect to a monotone circuit model $\mathcal{C}$ if for any $\mathcal{P}$-refutation there is a monotone circuit from $\mathcal{C}$ that computes an interpolant.

In the current work we consider algebraic proof systems that can be used more generally to prove that a given set of polynomial equations over a field does not have 0-1 solutions. 
This is more general than showing that a set of formulas in propositional logic is unsatisfiable (since propositional formulas can be encoded directly by polynomial equations). The notion of interpolation generalises naturally to this case, and we will prove our feasible interpolation result in this general setting. Namely, let $P_0(\vx,\vz)$ and $P_1(\vy,\vz)$ be two sets of polynomials in the sequences of variables $\vx$, $\vy$ and $\vz$ that are mutually unsatisfiable, i.e. that have no common (Boolean) root. A Boolean function $I(\vz)$ is an interpolant for $P_0$ and $P_1$ if it satisfies the following for every assignment $\valpha$:
\[
I(\valpha) = 
\begin{cases}
    0,\text{ if }P_0(\vx,\valpha)\text{ is satisfiable;}\\
    1,\text{ if }P_1(\vy,\valpha)\text{ is satisfiable.}
\end{cases}
\]
This is again equivalent to the following two conditions:
\begin{align}\label{eq:contrapositive-interpolation-polynomials}
&\text{ If } I(\valpha)=0\text{, then } P_1(\vy,\valpha) \text{ is unsatisfiable; and}\\\notag
&\text{ If } I(\valpha)=1\text{, then } \P_2(\vy,\valpha) \text{ is unsatisfiable}.
\end{align}
The notions of feasible and monotone feasible interpolation translate directly to the setting with general polynomial equalities.

\subsection{Read-once Oblivious Algebraic Branching Programs}

A variable ordering of a finite set of variables is a linear order on the set, e.g. $x_1 < \ldots < x_n$. A read-once oblivious algebraic branching program (roABP) \cite{Forbes14} in the variable ordering $x_1 < \ldots < x_n$ consists of a sequence of matrices $M_1,\ldots,M_n$ so that $M_i$ is a matrix whose entries are univariate polynomials in the variable $x_{i}$. The branching program computes the polynomial 
\[
\left(M_1\cdots M_n\right)_{(1,1)},
\]
i.e.~the upper left entry of the product of these matrices. This definition assumes that all the matrices here are of appropriate dimension so that the product is well-defined. The size measure for such a branching program is defined in terms of width: the 
\emph{width} of an roABP is the maximal dimension of the matrices.

Read-once oblivious algebraic branching programs can 
also be defined more combinatorially as a layered directed graphs with two special vertices, the source in the first layer and the sink in the last layer, and $n - 1$ layers in between. Each edge in the graph is between consecutive layers, that is directed from layer $i$ to layer $i+1$. The edges between layer $i$ and layer $i + 1$ are labelled with univariate polynomials in the variable $x_{i}$. Each source-to-sink path computes the polynomial that is the product of all the univariate polynomials on the path, and the program itself computes the sum of all the paths in the graph. It is not hard to see that these two definitions are equivalent. In this formulation, the width of the program corresponds to the maximum size of any layer.

There is a particular representation that is useful for our purposes below (cf. \cite{Forbes14}). If a polynomial $p$ is computable by a width $w$ roABP in a variables ordering $x_1 < \ldots < x_n$, then for any $i\in [n]$ it can be written as a sum
\[
\sum_{j\in [w]}p_{ij}q_{ij},
\]
where each $p_{ij}$ is a polynomial in the variables $x_1,\ldots,x_i$ and each $q_{ij}$ is a polynomial in the variables $x_{i + 1},\ldots,x_n$. 

\subsection{Algebraic Proof Systems}

A Nullstellensatz proof system, introduced to proof complexity in \cite{BeameIKPP96}, is based, as the name suggests, on Hilbert's Nullstellensatz. Here a Nullstellensatz proof of $p = 0$ from a sequence $p_1 = 0,\ldots,p_m = 0$ of polynomial equalities is a sequence $q_1,\ldots,q_m$ of polynomials so that 
\begin{equation}\label{ns-proof}
p = p_1q_1 + \ldots + p_mq_m.    
\end{equation}
A Nullstellensatz refutation of $p_1 = 0,\ldots,p_m = 0$ is a Nullstellensatz proof of $1 = 0$ from $p_1 = 0,\ldots,p_m = 0$. The complexity of these proofs is most commonly measured by the maximum degree of the summands, i.e. $\max_i\deg(p_iq_i)$, or by sparsity, or monomial-size, of the summands.

Polynomial Calculus \cite{CEI96} is a dynamic variant of Nullstellensatz, where the Nullstellensatz proof is derived by local inference rules. A Polynomial Calculus proof of $p = 0$ from a sequence $p_1 = 0,\ldots,p_m = 0$ of polynomial equalities is a sequence $q_1,\ldots,q_\ell$ of polynomials so that $q_\ell = p$ and one of the following holds for every $i\in [\ell]$:
\begin{enumerate}
    \item $q_i$ is one of the axioms, i.e. of the  polynomials $p_j$;
    \item there is $j < i$ and a variable $x$ so that $q_i = xq_j$;
    \item there are $j,k <i $ and field elements $a,b$ so that $q_i = aq_j + bq_k$.
\end{enumerate}
A Polynomial Calculus refutation a sequence $p_1 = 0,\ldots,p_m = 0$ is again a proof of $1 = 0$ from the sequence. We say that the proof is \emph{tree-like} if the underlying proof graph is a tree. Here the common complexity measures are again the degree and sparsity of the polynomials in the sequence. Another important measure for this work is the length of the proof, i.e. the number of polynomials in the sequence; above this is $\ell$.

Ideal Proof System (IPS) \cite{GP18} allows expressing the Nullstellensatz proof succinctly as a single algebraic circuit. Formally we have the following definition.
\begin{definition}[Ideal Proof System \cite{GP18}]\label{def:IPS}
    Let  $p_1 = 0,\ldots,p_m = 0$ be a sequence of polynomial equalities in variables $x_1,\ldots,x_n$. An IPS proof of $p= 0$ from $p_1 = 0,\ldots,p_m = 0$ is an algebraic circuit $C(\vx,\vz)$ in variables $x_1,\ldots,x_n$ and new placeholder variables $z_1,\ldots,z_m$ so that
    \begin{enumerate}
        \item $C(\vx,\bar{0}) = 0$;
        \item $C(\vx,p_1,\ldots,p_m) = p.$
    \end{enumerate}
    An IPS refutation of the sequence $p_1 = 0,\ldots,p_m = 0$ is an IPS proof of $1 = 0$ from the sequence. The size of the IPS proof is the circuit of $C$.
\end{definition}

Following \cite{FSTW21} we call an IPS proof \emph{linear} if the polynomial computed by the circuit $C$ computes a polynomial that is linear in the $\vz$ variables (\cite{GP18} calls this Hilbert-like IPS); the polynomial computed is thus of the form $\sum_{i\in [m]} q_i(\vx)z_i$ for some polynomials $q_i$. The size of this expression is roughly the size of the circuits computing the polynomials $q_i(\vx)$. Thus a linear IPS proof can be considered as a Nullstellensatz proof, where the size of the proof is measured by the algebraic circuit size of the polynomials $q_1,\ldots,q_m$ in \Cref{ns-proof}

In this work we consider linear IPS proofs, where the polynomials $q_1,\ldots,q_m$ in \Cref{ns-proof} are computed by read-once algebraic branching programs. We call these roABP-$\lIPS$ proofs. The complexity of an roABP-$\lIPS$ proof is measured by the maximum width of an roABP needed to compute any of the polynomials $q_1,\ldots,q_m$. We assume a common variable ordering for the whole proof and the complexity of the proof can and usually will depend on this choice. To emphasize the chosen variable ordering we talk about roABP-$\lIPS$ refutations in some particular variable ordering.

\subsubsection{A Translation from CNFs to Polynomials}

In order to consider these systems as a refutation systems for CNFs we fix a \emph{standard translation} of a set of clauses into a set of polynomial equations. For a clause $C = \bigvee_{i\in I}x_i\vee\bigvee_{j\in J}\neg x_j$ we define its translation to be the polynomial 
\[
\tr(C) := \prod_{i\in I}\left(1 - x_i\right)\prod_{j\in J}x_j
\]
and for a CNF $F$ in variables $x_1,\ldots,x_n$ consisting of clauses $C_1,\ldots,C_m$ we define its translation to be the set 
\[
\{\tr(C_i) = 0: i\in [m]\}\cup\{x_i^2 - x_i = 0 : i\in [n] \}.
\]

\subsection{Span Programs}

A span program \cite{KW93} consists of a labeled matrix $(M,\sigma)$ together with a target vector $t$, where $\sigma$ is a labeling of the rows, i.e. a mapping from the rows of $M$ to the literals in the variables $\vz$. We also allow labeling a row with the constant $1$. A span program is monotone, whenever there is no row labeled with a negative literal. The \emph{size} of the span program is the number of rows.
Equivalently define a span program as a set $U$ of pairs $(\ell,v)$, where $\ell$ is a $\vz$-literal or a constant $1$ and $v$ is a vector in some underlying vector space $V$ together with a target vector $t$ that lies also in the space. 

A span program as above \emph{computes} a Boolean function as follows: on a 0-1 assignment $a$ to the variables $\vz$, let $U(a)$ be the set of vectors $v$ so that there is some literal $\ell$ with $(\ell,v)\in U$ and $\ell(a) = 1$. In other words, $U(a)$ consists of all vectors picked by the assignment $\vz$, namely those vectors whose $\vz$-label gets 1 under $a$. The span program outputs $1$ if the target vector $t$ is in the span of the vectors in $U(a)$ and $0$ otherwise.

\begin{remark}
    Usually span programs are defined as labeled only with literals and not with the constant $1$. This inclusion is only for convenience as there are simple reductions to get rid of the constant.  If the span program $U$ is non-monotone, it suffices to replace each $(1,v)$ with the pairs $(z_i,v)$ and $(\bar{z}_i,v)$ for some variable $z_i$; the size of the program only increases by a factor of two. If the span program $U$ is monotone, replace each $(1,v)$ with the set of all the pairs $(z_i,v)$ for each variable $z_i$. The obtained program computes the same function unless the function is the constant function $1$ and the size increases only by a polynomial factor.
\end{remark}

%---------------------------------------------------
\section{Feasible Interpolation Property for  Fixed Variable Ordering}
%---------------------------------------------------
\label{sec:interpolation}

In this section we prove the feasible interpolation result for roABP-$\lIPS$ when the order of variables queried along the roABPs is fixed.

Before proceeding to the proof we show that we can assume a certain normal form for a given set of polynomial constraints. This normal form was introduced by Pudl\' ak and Sgall \cite{PS96-span-programs} (cf.~\cite{FGGR22} for more recent uses).

\begin{definition}[Monotone normal form]
Let $P(\vx,\vz)$ be a set of polynomials in the (pairwise disjoint) sets of variables $\vx,\vz$.
\begin{itemize}
\item 
$P(\vx,\vz)$ is \emph{in a $\vz$-normal form} if each polynomial in the set is of the form 
\begin{equation*}
p + z_ip', 
\end{equation*}
for some polynomials $p$ and $p'$ in the $\vx$ variables and  a \emph{single} $\vz$-variable $z_i$. 
\item
 $P(\vx,\vz)$ is \emph{monotone in $\vz$} if each polynomial in the set is of the form 
\begin{equation}\label{eq:monotone}
    \left(\prod_{i\in I} z_i\right) \cd p,
\end{equation}
where $I$ is some subset of indices and $p$ is a polynomial in the $\vx$ variables only.
Here, monotonicity is with respect to the $\vz$-variables (and not necessarily the $\vx$-variables), meaning that over 0-1 values to the $\vz$-variables, the polynomial is monotone in the following sense: when $p$ is fixed and has a nonnegative value, flipping a $z_i$ from 0 to 1 can only increase the value of \eqref{eq:monotone}.

\item $P(\vx,\vz)$ is in a \emph{monotone $\vz$-normal form} if the size of the set of indices $I$ in \eqref{eq:monotone}  is at most $1$ for every polynomial in the set.
In other words, each of its polynomials is of the form $z_i p'$ or $p'$.  
\end{itemize}
\end{definition}

In the below, $\operatorname{sparsity}(p)$, for a polynomial $p$, denotes the number of monomials in $p$, and we also call $\operatorname{sparsity}(p)$ the \emph{sparse size of $p$}.

\begin{lemma}[Normal form lemma]
\label{lem:normal-form}
Let $P(\vx,\vz)$ be a finite set of sparse polynomials over a field $\F$.
Let $s := \sum_{p\in P}\operatorname{sparsity}(p)$ and $n := |\vx|+|\vz|$, and assume that every polynomial in $P$ has degree at most $(s+n)^{O(1)}$.
Then, there is a set $P'(\vx,\vw,\vz)$ of polynomials, of sparsity $\poly(s+n)$ and where $\vw$ are fresh variables, such that:
\begin{enumerate}
\item $P'$ is in $\vz$-normal form, namely every polynomial in $P'$ is of the
form $q(\vx,\vw)+z_i q'(\vx,\vw)$, for some $i$, or contains no $\vz$-variable;

\item \label{it:normal-form-equis} for every assignment $\alpha\in\{0,1\}^{|\vz|}$, $P$ and $P'$ are equisatisfiable in the following sense:
\[
\exists \vx\in\{0,1\}^{|\vx|}\; P(\vx,\alpha)=0
\quad\Longleftrightarrow\quad
\exists \vx\in\{0,1\}^{|\vx|}\exists \vw\in\F^{|\vw|}\; P'(\vx,\vw,\alpha)=0;
\]

\item every polynomial $p\in P$ has an $\roABP$-\lIPS\ derivation from
$P'$ of width $\poly(s+n)$, given any variable ordering.
\end{enumerate}

Moreover, if $P$ is monotone in $\vz$, then $P'$ can be chosen in monotone
$\vz$-normal form.
\end{lemma}

\begin{remark}
    The auxiliary variables $\vw$ introduced in the lemma are algebraic extension
variables; they are not constrained by Boolean axioms. Thus satisfiability of
$P'(\vx,\vw,\alpha)$ means satisfiability by Boolean values for the original
variables $\vx$ and field values for the auxiliary variables $\vw$.
\end{remark}

\begin{proof}
    We give two different constructions depending on whether $P$ is monotone in $\vz$ or not.

\paragraph{Nonmonotone case.}
Assume first that $P$ is not necessarily monotone in $\vz$. Introduce a fresh
auxiliary variable $w_i$ for every variable $z_i$. For every polynomial
$p(\vx,\vz)\in P$, let $p^\star(\vx,\vw)$ be obtained from $p$ by replacing
each $z_i$ by $w_i$. Define
\[
        P'(\vx,\vw,\vz)
        :=
        \{p^\star(\vx,\vw) : p\in P\}
        \cup
        \{z_i-w_i : i\in [|\vz|]\}.
\]
Then $P'$ is in $\vz$-normal form: each $p^\star$ contains no $\vz$-variable,
and each $z_i-w_i$ is of the form $-w_i+z_i$.

The equisatisfiability condition is immediate. Indeed, after fixing an
assignment $\alpha\in\{0,1\}^{|\vz|}$, the equations $z_i-w_i=0$ force
$w_i=\alpha_i$ for every $i$. Hence
\[
        p^\star(\vx,\vw)=0 \text{ for all } p\in P
\]
is equivalent to
\[
        p(\vx,\alpha)=0 \text{ for all } p\in P.
\]
Thus
\begin{equation}\label{eq:750}
    \exists \vx\in\{0,1\}^{|\vx|}\; P(\vx,\alpha)=0
    \quad\Longleftrightarrow\quad
    \exists \vx\in\{0,1\}^{|\vx|}\exists \vw\in\F^{|\vw|}\;
    P'(\vx,\vw,\alpha)=0.
\end{equation}

It remains to derive every original polynomial $p\in P$ from $P'$ by a small $\roABP$-\lIPS\ derivation.
This is done using routine variable-by-variable substitution via a telescoping identity, as follows.

Since $p^\star$ is already an axiom of $P'$, it
suffices to derive $p-p^\star$ from the equations $z_i-w_i$. For a monomial
\[
        m(\vx,\vz)=c\vx^\gamma\prod_{j=1}^k z_{i_j}
\]
of $p$, where repetitions among the indices $i_j$ are allowed, write
\[
        m^\star(\vx,\vw)=c\vx^\gamma\prod_{j=1}^k w_{i_j}.
\]
Then the telescoping identity is
\[
\prod_{j=1}^k z_{i_j}-\prod_{j=1}^k w_{i_j}
=
\sum_{j=1}^k
\left(\prod_{h<j} z_{i_h}\right)
\left(\prod_{g>j} w_{i_g}\right)
(z_{i_j}-w_{i_j}).
\]
Multiplying by $c\vx^\gamma$ gives a derivation of $m-m^\star$ from the equations $z_i-w_i$. Summing over all monomials of $p$ gives a derivation of $p-p^\star$, and hence of $p$, from $P'$.

Summing over all monomials of $p$ gives a derivation of $p-p^\star$, and hence of $p$, from $P'$. For each equation $z_i-w_i$, the corresponding coefficient is a sum of at most $\operatorname{sparsity}(p)\cdot\deg(p)$ monomials from the telescopic identity. A monomial has width-$1$ roABP in any variable ordering, and a sum of $T$ monomials has roABP width at most $T$ by taking the  sum of the width-$1$ roABPs. Hence every coefficient in this derivation has roABP width
at most $\operatorname{sparsity}(p)\cdot\deg(p)$, which is
$\poly(s+n)$ by assumption. Therefore each $p\in P$ has an $\roABP$-\lIPS\ derivation from $P'$ of width $\poly(s+n)$ in any variable ordering.

\paragraph{Monotone case.}
Assume that $P$ is monotone in $\vz$. Thus every polynomial $p\in P$ is of
the form
\[
        p(\vx,\vz)=\left(\prod_{i\in I_p} z_i\right)p'(\vx),
\]
where $p'(\vx)\in\F[\vx]$. If $I_p=\emptyset$, we keep $p=p'(\vx)$ unchanged.
Otherwise, for every $i\in I_p$ introduce a fresh auxiliary variable
$w_{p,i}$, and replace $p$ by the polynomials
\[
        p'(\vx)+\sum_{i\in I_p} w_{p,i},
        \qquad
        z_iw_{p,i}\quad(i\in I_p).
\]
Let $P'(\vx,\vw,\vz)$ be the union of these polynomials over all $p\in P$.
Then $P'$ is in monotone $\vz$-normal form.

The equisatisfiability condition is checked after fixing an arbitrary
assignment $\alpha\in\{0,1\}^{|\vz|}$ to the $\vz$-variables. Consider one
polynomial
\[
        p(\vx,\vz)=\left(\prod_{i\in I_p}z_i\right)p'(\vx)
\]
and the normalised equations associated with it:
\[
        p'(\vx)+\sum_{i\in I_p}w_{p,i}=0,
        \qquad
        z_iw_{p,i}=0 \quad (i\in I_p).
\]
We are trying to prove the following: for every fixed Boolean assignment $\alpha$ to the $\vz$-variables,
$
P(\vx,\valpha)=0 $
is satisfiable in Boolean 
$\vx$ if and only if $P'(\vx,\vw,\valpha)=0$  is satisfiable in Boolean $\vx$  and field-valued $\vw$.
So fix a Boolean assignment $\valpha$ to $\vx$. After substituting $\vz=\valpha$, the original equation becomes
\[
        \left(\prod_{i\in I_p}\alpha_i\right)p'(\vx)=0.
\]
If all $\alpha_i=1$ for $i\in I_p$, this is just $p'(\vx)=0$. On the
normalised side, the equations $\alpha_iw_{p,i}=0$ force all $w_{p,i}=0$,
and hence the first normalised equation also becomes $p'(\vx)=0$.

If some $\alpha_{i_0}=0$, then the original equation becomes
$0\cdot p'(\vx)=0$, and hence imposes no condition on $\vx$. The normalised
equations impose no condition on $\vx$ either: for any assignment to $\vx$, set
$w_{p,i_0}=-p'(\vx)$ and set all other $w_{p,i}$'s to $0$. Then
$p'(\vx)+\sum_{i\in I_p}w_{p,i}=0$, and all equations
$\alpha_iw_{p,i}=0$ hold. Thus, for each fixed $\alpha$, the normalised
equations associated with $p$ impose exactly the same condition on $\vx$ as
the original equation $p(\vx,\alpha)=0$.

Since the auxiliary variables used for different polynomials $p\in P$ are
disjoint, the same argument applies independently to all polynomials in $P$.
Therefore \Cref{eq:750}
holds.
\medskip

It remains to derive each original polynomial $p$ from $P'$ by a small
$\roABP$-$\lIPS$ derivation. For
$p=(\prod_{i\in I_p}z_i)p'(\vx)$, we have the identity
\begin{equation}\label{eq:859}
\left(\prod_{i\in I_p}z_i\right)p'(\vx)
=
\left(\prod_{i\in I_p}z_i\right)
\left(p'(\vx)+\sum_{i\in I_p}w_{p,i}\right)
-
\sum_{i\in I_p}
\left(\prod_{j\in I_p\setminus\{i\}}z_j\right)
(z_iw_{p,i}).
\end{equation}
Thus $p$ is obtained as an $\roABP$-\lIPS\ linear combination of the
polynomials introduced for $p$ whose coefficient polynomials are monomials
in the $\vz$-variables. Hence these coefficients have width-$1$ roABPs in
every variable ordering. In particular, the width of the derivation, measured
by the coefficient roABPs, is constant.
 Therefore every $p\in P$ has an
$\roABP$-$\lIPS$ derivation from $P'$ of width $\poly(s+n)$ (in fact, of constant width).

Finally, the total sparse size of $P'$ is polynomial in $s+n$: for each
$p$, the construction introduces one polynomial of sparsity at most
$\operatorname{sparsity}(p)+|I_p|$ and $|I_p|$ monomial equations, and
$|I_p|\leq\deg(p)\leq(s+n)^{O(1)}$.
\end{proof}

The following corollary sets how to apply \Cref{lem:normal-form}.
The argument is a simple use of  composition of proofs, while making proofs are not blowing up in roABP width when the order of variables is kept $\vx\cup\vw < \vz\cup\vy$.

\begin{corollary}[Normal form application]
\label{cor:normal-form-application}
Let $P_0(\vx,\vz)$ and $P_1(\vy,\vz)$ be sets of polynomial equations, where
$\vx,\vy,\vz$ are pairwise disjoint sets of variables. Assume that $P_0$ is a
set of polynomials of both degree and sparsity $\poly(|\vx+\vz|)$.\footnote{For example, $P_0$ may be the standard arithmetization of
a $k$-CNF, for a constant $k$.} Suppose that $P_0\cup P_1$ has an
\textup{roABP}-$\lIPS$ refutation of width $w$ in a variable ordering in which
the $\vx$-variables precede all other variables.
Then, there is a set of polynomial equalities $P'_0(\vx',\vz)$, where
\begin{itemize}
    \item $\vx',\vy,\vz$ are pairwise disjoint sets of variables.
    \item $P'_0$ is in $\vz$-normal form and has sparsity $\poly(|\vx+\vz|)$.
    \item For every Boolean assignment $\alpha$ to the $\vz$-variables,
        \[
            P_0(\vx,\valpha)\text{ is satisfiable}
            \quad\Longleftrightarrow\quad
            P'_0(\vx',\valpha)\text{ is satisfiable}.
        \]
        Consequently, the pairs $(P_0,P_1)$ and $(P'_0,P_1)$ define the same
        interpolation problem, and hence have the same interpolant.
    \item $P'_0\cup P_1$ has an \textup{roABP}-$\lIPS$ refutation of width
    $\poly(w,|P_0|)$ in a variable ordering in which the $\vx'$-variables precede both the $\vy$- and $\vz$-variables.
\end{itemize}
Moreover, if $P_0$ is monotone in the $\vz$-variables then $P'_0$ is in monotone $\vz$-normal form. 
\end{corollary}

\begin{proof}
Apply \Cref{lem:normal-form} to $P_0(\vx,\vz)$. Let
$P'_0(\vx,\vw,\vz)$
be the resulting set of polynomial equations, and write
$\vx' := (\vx,\vw)$.
By \Cref{lem:normal-form}, the set $P'_0$ is in (monotone) $\vz$-normal form, has
sparse size polynomial in $|\vx|+|\vz|$, and for every assignment to the common variables $\vz$, the left-hand side
of the interpolation problem (i.e., $P_0$ and $P_0'$) has the same satisfiability status for
$P_0$ and for $P'_0$. Consequently, the pairs $(P_0,P_1)$ and
$(P'_0,P_1)$ define the same interpolation problem, and hence have the same interpolant.

It remains to show that the refutation of $P_0\cup P_1$ can be converted into
a refutation of $P'_0\cup P_1$ with only a polynomial increase in width and in an ordering in which $\vx'$-variables come before all other variables. Let
\begin{equation}
    \sum_j R_j(\vx,\vz,\vy)\, P_{0,j}(\vx,\vz)
    +
    \sum_t L_t(\vx,\vz,\vy)\, P_{1,t}(\vy,\vz) =1
    \label{eq:8.11}
\end{equation}
be the assumed $\roABP$-$\lIPS$ refutation of $P_0\cup P_1$ of width
$w$, where the $P_{0,j}$ enumerate the polynomials of $P_0$ and the $P_{1,t}$
enumerate the polynomials of $P_1$.

By ~\Cref{lem:normal-form}, for every $j$ there is an
$\roABP$-$\lIPS$ derivation of $P_{0,j}$ from $P'_0$ of width
$\poly(|\vx|+|\vz|)$. Thus we may write
\begin{equation}
P_{0,j}(\vx,\vz)
=
\sum_i Q_{i,j}(\vx,\vw,\vz)\,P'_{0,i}(\vx,\vw,\vz),
\label{eq:8.21}
\end{equation}
where each coefficient $Q_{i,j}$ has an $\roABP$ of width
$\poly(|\vx|+|\vz|)$ in any variable ordering.

We now substitute the identities~\eqref{eq:8.21} into the refutation~\eqref{eq:8.11}. This gives
\begin{equation}
\begin{aligned}
1
&=
\sum_j R_j(\vx,\vz,\vy)
   \left(\sum_i Q_{i,j}(\vx,\vw,\vz)\,P'_{0,i}(\vx,\vw,\vz)\right)
+
\sum_t L_t(\vx,\vz,\vy)\,P_{1,t}(\vy,\vz)  \\
&=
\sum_i
\left(\sum_j R_j(\vx,\vz,\vy)\,Q_{i,j}(\vx,\vw,\vz)\right)
P'_{0,i}(\vx,\vw,\vz)
+
\sum_t L_t(\vx,\vz,\vy)\,P_{1,t}(\vy,\vz).
\end{aligned}
\label{eq:8.31}
\end{equation}
This is a linear IPS refutation of $P'_0\cup P_1$.

It remains only to check the width and variable ordering. Extend the original variable ordering so
that all variables in $\vx'=(\vx,\vw)$ precede the variables in $\vy,\vz$.
The original coefficients $R_j$ and $L_t$ can be viewed as polynomials in the
larger variable set by ignoring the new $\vw$-variables; this does not increase
their roABP width. The coefficients $Q_{i,j}$ have polynomial-width roABPs in
this same extended ordering by \Cref{lem:normal-form}. Products of two
roABPs in the same variable ordering can be computed (while preserving the varaible ordering) by taking tensor products
of the layers, so the width multiplies\footnote{We use the standard closure properties of roABPs in a fixed variable order. Suppose $P$ and $Q$ are computed in the same variable order $x_1<\cdots<x_n$, with matrix presentations $P=(A_1\cdots A_n)_{1,1}$ and $Q=(B_1\cdots B_n)_{1,1}$, where $A_i$ and $B_i$ are the matrices for the $x_i$-layer. For each $i$, define the $x_i$-layer of the product roABP to be $C_i=A_i\otimes B_i$. Then,
$
C_1\cdots C_n=(A_1\otimes B_1)\cdots(A_n\otimes B_n)=(A_1\cdots A_n)\otimes(B_1\cdots B_n)
$, where the rightmost equality is by basic tensor calculus.
Hence the output entry of the product roABP equals $(A_1\cdots A_n)_{1,1}(B_1\cdots B_n)_{1,1}=PQ$. If the widths of the two roABPs are $w_P$ and $w_Q$, then the matrices $C_i$ have dimension at most $w_Pw_Q$, so the width is at most $w_Pw_Q$.}. Sums of polynomially many such roABPs
can be computed by taking their direct sum (i.e., by putting them side by side in parallel), so the width increases by at most
a polynomial factor. Hence each coefficient
\[
\sum_j R_j Q_{i,j}
\]
in~\eqref{eq:8.31} has roABP width $\poly(w,|P_0|)$, and the coefficients $L_t$ retain
width at most $w$ in a variable order in which $\vx'$-variables precede both the $\vy$- and
$\vz$-variables. Therefore, $P'_0\cup P_1$ has an
$\roABP$-$\lIPS$ refutation of width $\poly(w,|P_0|)$ in this variable order.
\end{proof}

We are now ready to prove the feasible interpolation theorem.
\begin{theorem}[Feasible interpolation for roABP-$\lIPS$]\label{thm:feasible-interpolation}
    Let $P_0(\vx,\vz)$ and $P_1(\vy,\vz)$ be sets of polynomial equalities, where $\vx,\vy$ and $\vz$ are pairwise disjoint variables and $P_0$ is a set of polynomials of both degree and sparsity $\poly(|\vx|+|\vz|)$. Assume that $P_0$ is in $\vz$-normal form. Suppose that $P_0\cup P_1$ has a \textup{roABP}-$\lIPS$ refutation of width $w$ in a variable ordering where $\vx$ variables precede all other variables. Then there is a span program of size $\poly(w\cdot |P_0(\vx,\vz)|)$ that computes an interpolant for $P_0$ and $P_1$. Moreover, if all the polynomials in $P_0$ are in monotone $\vz$-normal form, then the span program is monotone. 
\end{theorem}

\begin{proof}
    By \Cref{cor:normal-form-application} we can assume without loss of generality that $P_0$ is in a (monotone) normal form and let  
    \begin{equation}\label{given-refutation}
        \sum_{p\in P_0} a_p p + \sum_{q\in P_1}b_q q = 1
    \end{equation}
    be a refutation of $P_0\cup P_1$. 
    
    Note that, by the assumption that for every $p\in P_0$ the polynomial $a_p$ is computable by a width $w$ read-once oblivious algebraic branching program, it can be written in the form
  \begin{equation}\label{eq:decompose-ap}
        a_p= \sum_{i\in [w]}a_{p,i}(\vx)\cd r_i(\vy,\vz),
  \end{equation}
for some polynomials $a_{p,i}$ in the $\vx$ variables only and some polynomials $r_i$ in the $\vy$ and $\vz$ variables only. Indeed, consider the layer of the roABP immediately after the last $\vx$-variable in the given variable ordering (recall that the $i$-th layer in a roABP contains only edges labeled with univariate polynomials in the $i$-th variable in the variable linear ordering). Since the roABP has width $w$, this layer contains at most $w$ nodes. For each node $i \in [w]$, let $a_{p,i}$ denote the polynomial computed from the source to node $i$, and let $r_i$ denote the polynomial computed from node $i$ to the sink. Then $a_{p,i}$ depends only on the $\vx$-variables, while $r_i$ depends only on the remaining variables, namely $\vy,\vz$. Summing over all nodes in this layer gives the polynomial $a_p$. If the layer has fewer than $w$ nodes, we pad the decomposition by setting the missing $a_{p,i}$'s equal to $0$.
\medskip

We now construct the span program $I$ as follows: 

\begin{enumerate}
\item \label{it:z-labels}
For any polynomial $p\in P_0$ of the form $p' + z_j p''$ add to the span program the entries
    \[
    \left(z_j, a_{p,i}(p' + p'')\right)\text{ and }(1 - z_j, a_{p,i}p'),\text{ for all }i\in [w].
    \]
Thus, if $z_j=1$ the span program picks the vector $a_{p,i}(p' + z_jp'')=a_{p,i}(p' + p'')$; while if $z_j=0$ the span program picks the vector $a_{p,i}(p' + z_j p'')=a_{p,i}p'$. 
\item For any polynomial $p\in P_0$ of the form $p'$ (i.e., no $\vz$-variables appear in the polynomial)
add to the span program the entries
    \[
    (1, a_{p,i}p'), \text{ for all }i\in [w].
    \]

\item The target vector is the constant polynomial $1$. 
\end{enumerate}   
Note that if $P_0$ is monotone in the $\vz$-variables then this span program is monotone, because in this case $p\in P_0$\ is of the form $p'+z_j p''$ with $p'=0$, hence $(1 - z_j, a_{p,i}p') = (1 - z_j, 0)$ can be discarded from the span program.

\smallskip
   
We claim that the constructed span program computes an interpolant for $P_0$ and $P_1$, namely, the conditions in \eqref{eq:contrapositive-interpolation-polynomials}
 are met: if $I(\valpha)=1$ then $P_0(\vx,\valpha)$ is unsatisfiable; and if $I(\valpha)=0$ then $P_1(\vy,\valpha)$ is unsatisfiable.

We consider the two cases separately. 
Suppose first that the span program $I$ outputs $1$ on an input $\valpha$. The program outputs $1$, when the constant polynomial $1$ is in the linear span of the vectors selected by $\valpha$. Note that by \Cref{it:z-labels} above these vectors correspond to polynomials in $P_0(\vx,\valpha)$ multiplied by the polynomials $a_{p,i}$. Thus expressing $1$ as a linear span of these polynomials constitutes a Nullstellensatz refutation of $P_0(\vx,\alpha)$, and so $P_0(\vx,\alpha)$ is unsatisfiable.

    Otherwise, suppose  that the span program $I$ outputs $0$ on an input $\valpha$. Suppose towards a contradiction that  $P_1(\vy,\valpha)$ is  satisfiable and let $\vbeta$ be an assignment to the $\vy$ variables that satisfies $P_1(\vy,\valpha)$ (so $P_1$ vanishes under $\valpha,\vbeta$). Applying both $\valpha$ and $\vbeta$ to the given refutation \eqref{given-refutation} it simplifies to a refutation of $P_0(\vx,\valpha)$:
    \[
    \sum_{p\in P_0}a_p(\vx,\vbeta,\valpha)p(\vx,\valpha) = 1.
    \]
    Moreover, by \Cref{eq:decompose-ap}, for every $p\in P_0$ the polynomial $a_p(\vx,\vbeta,\valpha)$ is a linear combination of the polynomials $a_{p,i}(\vx)$, for $i\in [w]$:
\begin{equation}\label{eq:crux}
        a_p(\vx,\vbeta,\valpha)= \sum_{i\in [w]}a_{p,i}(\vx)\cd r_i(\vbeta,\valpha).
  \end{equation}    
    The polynomials $a_p(\vx,\vbeta,\valpha)p(\vx,\valpha)$ are thus linear combinations of the vectors selected by $\valpha$, and the span program should have output $1$, concluding the proof.
    
        Note that \Cref{eq:crux} is the crucial equality allowing us to conclude the theorem: for this equality to hold we must be able to write the left hand side as a \emph{linear combination} of $a_{p,i}(\vx)$, for $i\in[x]$. This is only possible because the variables $\vy,\vz$ come  \emph{after} the \vx-variables in the variable ordering, so that the $r_i(\vbeta,\valpha)$ become constant.
\end{proof}

\begin{remark}
    The statement of \Cref{thm:feasible-interpolation} in the monotone case can be applied to a slightly more general class of formulas. This is done by observing that the monotone case in \cref{lem:normal-form} does not depend on sparsity in the sense that the monotone feasible interpolation in \Cref{thm:feasible-interpolation} is applicable   for every monotone $P_0$ so that the polynomials in it are computable by low-width roABPs.
\end{remark}

%-----------------------------------------------------
\subsection{Application: Lower Bounds for Fixed Variable Ordering}
%----------------------------------------------------
\label{sec:fixed-order}

In this section we use the monotone feasible interpolation theorem of the previous section to obtain lower bounds against roABP-$\lIPS$. For this we will require the following lower bounds for monotone span programs.

\begin{definition}[$\GEN_n$ function hard for monotone span programs \cite{PitassiRobere2018}]
\label{def:GEN}
For a positive integer $n$, let $T \subseteq [n]^3$ be a set of triples. We say that $T$ \emph{generates} a point $w \in [n]$ if either $w=1$, or there exists a triple $(u,v,w)\in T$ such that $T$ generates both $u$ and $v$. The function $\GEN_n$ takes as input the characteristic vector of a set $T \subseteq [n]^3$ and outputs $1$ if and only if $T$ generates the point $n$. 
\end{definition}

Equivalently, starting from the initially generated point 1, repeatedly add a point $w$ whenever there is a triple $(u,v,w)\in T$ such that both $u$ and $v$ have already been generated. Then $\GEN_n(T)=1$ if and only if this process eventually generates $n$.

\begin{theorem}[\cite{robere_exponential_2016,PitassiRobere2018}]\label{thm:spanprog-lb}
  The monotone function $\GEN_n$ is computable in polynomial time. Furthermore,  over any field any monotone span program computing $f$ requires size $2^{n^{\Omega(1)}}$.
\end{theorem}

The following lemma gives a way to transform monotone functions computable by small circuits into monotone propositional formulas describing the yes and no instances of the function respectively. A similar construction appeared in \cite{robere_exponential_2016} starting from a monotone circuit. The lemma below is a non-uniform analogue of the general construction of an unsatisfiable split formula from a disjoint $\NP$ pair \cite{Kra97-Interpolation}. We include the proof for completeness, and to note that the circuit itself need not be monotone as long as the function is.

\begin{lemma}\label{lem:function-to-CNF}
    For any monotone function $f\in \NP/\poly\cap\coNP/\poly$ there are polynomial sized $3$-CNF formulas  $\phi_0(\vx,\vz)$ and $\phi_1(\vy,\vz)$ so that 
    \begin{itemize}
        \item $f(\valpha) = 0$ if and only if $\phi_0(\vx,\valpha)$ is satisfiable;
        \item $f(\valpha) = 1$ if and only if $\phi_1(\vy,\valpha)$ is satisfiable;
        \item the $\vz$-variables appear in $\phi_0(\vx,\vz)$ only negatively and in $\phi_1(\vy,\vz)$ only positively.
    \end{itemize}
\end{lemma}

\begin{proof}
    We construct first the formula $\phi_0$. For this let $C(\vz,\vw)$ be a polynomial size nondeterministic circuit computing the negation of $f$, where $\vz$ are the input variables and $\vw$ the nondeterministic variables; its existence is guaranteed by the fact that $f\in\coNP/\poly$. To define $\phi_0$ introduce for each gate $g$ of the circuit $C$ a new variable $x_g$, and construct $\phi_0$ as follows:
    \begin{itemize}
        \item if $g$ is a $\circ$\kern 0.09em-gate, where $\circ\in\{\wedge,\vee\}$ with children $g_0$ and $g_1$, include in $\phi_0$ the clauses encoding the following formula  $x_g\leftrightarrow (x_{g_0}\circ x_{g_1});$ 
        \item if $g$ is a $\neg$\kern 0.09em-gate with a child $g_0$, include in $\phi_0$ the clauses encoding the formula $x_g\leftrightarrow\neg x_{g_0}$;
        \item for the output gate $g$, include in $\phi_0$ the clause $x_{g}$;
        \item finally for each input gate $g$ labeled with a variable $z_i$, include in $\phi_0$ the clause $\neg z_i\vee x_g$.
    \end{itemize}
    The constructed CNF satisfies the last condition. We will prove the first condition. Suppose first that $f(\alpha) = 0$. Then there is some assignment $\beta$ to the $\vw$ variables so that $C(\alpha,\beta) = 1$. We can satisfy $\phi_0(\vx,\alpha)$ by assigning the input gate variables $x_g$ according to the assignments $\alpha$ and $\beta$ and the intermediate variables so that they respect the required conditions. Suppose then that $\phi(\vx,\alpha)$ is satisfiable, and let $\gamma$ be the assignment to the gate variables corresponding to gates labeled with $\vz$ variables. By rewiring we may assume that there is a single input gate labeled with any $\vz$ variables, and thus $|\gamma| = |\alpha|$. As $\gamma$ is a part of a satisfying assignment of $\phi_0(\vx,\alpha)$ we have that $f(\gamma) = 0$. Also $\alpha\leq\gamma$ as $\phi_0$ contains the clauses $\neg z_i\vee x_g$. Hence $f(\alpha)\leq f(\gamma) = 0$ as $f$ is monotone.

    The formula $\phi_1(\vy,\vz)$ is constructed similarly using the nondeterministic circuit for $f$ itself with the modification that for input gate $g$ labeled with variable $z_i$, we include in $\phi_1$ the clause $z_i\vee\neg y_g$.
\end{proof}

With \cref{thm:feasible-interpolation} and the existing monotone span program lower bounds of \cite{robere_exponential_2016,PitassiRobere2018} we are ready to prove the main theorem of this section.

\begin{theorem}\label{thm:fixed-order-lbs}
    There is an unsatisfiable $3$-CNF $\phi_0(\vx,\vy)\wedge\phi_1(\vy,\vz)$ that requires \textup{roABP}-$\lIPS$ refutations of width $2^{N^{\Omega(1)}}$ in any variable ordering where $\vx$ variables precede all the other variables, where $N$ is the number of variables in the formula.
\end{theorem}

\begin{proof}
    Consider the function $\GEN_n$, and note that it is indeed computable in polynomial time by iteratively closing the initially generated set $\{1\}$ under the triples in $T$. In particular, it belongs to $\NP/\poly \cap \coNP/\poly$. Thus, let $\phi_0(\vx,\vz)$ and $\phi_1(\vy,\vz)$ be the formulas encoding the no and yes instances of the function $\GEN_n$, respectively, given by \cref{lem:function-to-CNF}.
    The obtained $3$-CNF has $\poly(n)$ variables and $\poly(n)$ clauses, and all the $\vz$ variables appear only negatively in $\phi_0(\vx,\vz)$.
     
    Suppose there is an roABP-$\lIPS$ refutation of $\phi_0(\vx,\vz)\wedge\phi_1(\vy,\vz)$ of width $w$ in some variable ordering where $\vz$ variables precede the $\vy$ and $\vz$ variables. Note that the standard polynomial translation of $\phi_0(\vx,\vz)$ is in monotone $\vz$-normal form. Thus, by \cref{thm:feasible-interpolation} there is a size $\poly(n\cd w)$ monotone span program computing an interpolant $I$ for $\phi_0(\vx,\vz)\wedge\phi_1(\vy,\vz)$.

    To finish the proof it suffices to note that the computed interpolant is exactly the function $\GEN_n$. 
    % on $n$ input bits. 
    To see this note that the interpolant $I$ satisfies
    \[
I(\valpha) =
\begin{cases}
    0, \text{ if }\phi_0(\vx,\valpha)\text{ is satisfiable;}\\
    1, \text{ if }\phi_1(\vy,\valpha)\text{ is satisfiable;}
\end{cases}
\]
and so by the construction of $\phi_0(\vx,\vz)$ and $\phi_1(\vy,\vz)$ it satisfies
    \[
I(\valpha) =
\begin{cases}
    0, \text{ if } \GEN_n(\valpha) = 0\\
    1, \text{ if } \GEN_n(\valpha) = 1.
\end{cases}
\]
Thus the function has size $\poly(n\cd w)$ monotone span programs, and so, by \Cref{thm:spanprog-lb} we have that
$w\geq 2^{n^{\Omega(1)}}$. As $N = \poly(n)$ it follows that $w\geq 2^{N^{\Omega(1)}}$.
\end{proof}

\section{Lower Bounds for Arbitrary Variable Ordering}
\label{sec:any-order-roABP}

In this section, we build a single family of unsatisfiable CNF formulas, such that each CNF  requires large roABP-$\lIPS$ refutations in every variable ordering. A natural strategy that has been employed in such settings is to use the hard instance defined in the previous section, which we showed is hard for a fixed variable ordering, and then \emph{lift} that instance to obtain a lower bound for all variable orderings. 
 
\para{Background on lifting in algebraic proof complexity.} The first roABP-$\lIPS$ lower bound for all variable order was proved in~\cite{FSTW21} for an algebraic instance built in the subset sum. In their work, they first prove that the following unsatisfiable algebraic formula  (a polynomial equation) is hard to refute by an roABP-$\lIPS$ refutation in a fixed variable ordering, namely $\vu < \vv$. 
$$f(\vu, \vv):  \sum_{i \in [n]} u_iv_i - \beta = 0,$$
where $\beta \notin \{0,1,\ldots, n\}$. Next, they construct another hard instance, using two new sets of variables, namely $\vz = \{z_{i,j}\}$ for ${i,j \in [{2n}]}$ and $\{x_i\}$ for ${i \in [{2n}]}$. Specifically, they create the following instance:
$$g(\vz, \vx): \sum_{i<j} z_{i,j} x_i x_j\ - \beta =0,$$
where $\beta \notin \{0,1, \ldots, {n \choose 2}\}$. 
The main idea is that for any order $\sigma$ of the $\vx$ variables, there exists a Boolean assignment to the $\vz$ variables and a projection of $\vx$ variables to $\vu$ and $\vv$ variables that recovers the original hard instance. Hence, if the original instance can be obtained as a projection of the new instance, then the lower bound proved for the original instance carries over to the new instance. They call $g(\vz,\vx)$ as the \emph{lifted} instance of $f(\vu, \vv)$. This lifting idea is used by other $\IPS$ lower bounds such as~\cite{GHT22, HLT24, EGLT26}.

In all the above works, the input instances were purely algebraic and not a CNF formula. 

\smallskip 

It is worth noting that the question of proving lower bounds for all variable orders naturally occurs in the context of OBDD-based proof systems, which are  propositional proof systems operating with ordered binary decision diagrams (introduced by Atserias, Kolaitis and Vardi~\cite{AKV04}). This is addressed in the works of Kraj\'{i}\v{c}ek~\cite{Kra2007}, Segerlind~\cite{Seg2007}, and a recent work of~\cite{BIKRS21}. While there are some similarities between our proof strategy and that of~\cite{Kra2007}, the arguments are different.

\subsection{The Hard Instance}

Here we describe our construction of the hard instance.
We start by fixing some notation.

Let us call  the hard 3-CNF formula from the previous section $\phi = \phi_0 \wedge \phi_1$. From this formula, we will build a new hard instance, which we call $\Psi$. Let us denote the variables of the original instance $\phi$ by $\vx = \{x_1, x_2, \ldots, x_N\}$ and assume that it is hard in the variable ordering $x_1 < \ldots < x_N$. Let the clauses in $\phi$ be denoted by $C_1, C_2, \ldots, C_M$. A specific clause $C$ in this formula is a conjunction of at most three literals, i.e., $C = (\ell_{1} \vee \ell_{2} \vee \ell_{3})$. Moreover, for a literal $\ell$ in $C$, let $\var(\ell)$ denote the variable corresponding to the literal $\ell$ and $\sign(\ell)$ denote the sign of that literal. 
\bigskip

The hard instance $\Psi$ we construct has two sets of variables denoted $\vu$ and $\vv$, where $|\vv|=N$ and $|\vu|=O(M \cd N^3)$.
Our construction guarantees the following properties. 

\begin{itemize}
    \item $\Psi(\vu,\vv) = \Phi_1(\vu,\vv) \wedge \Phi_2(\vu)$, where $\Phi_1$ is a CNF over the variables $\vu, \vv$ and $\Phi_2$ is a CNF over $\vu$ variables.
    \item Additionally, $\Psi(\vu,\vv)$ is unsatisfiable. Specifically, for every Boolean assignment $\valpha$ to $\vu$ exactly one of the following two holds:
    \begin{itemize}
        \item either $\Phi_2(\valpha)$ is not satisfied;
        \item or $\Phi_2(\valpha)$ is satisfied and there exists a relabelling of $\vv$ variables by $\vx$ variables such that $\Phi_1(\valpha,\vv)$ reduces to $\phi$, the hard instance for the fixed variable order (hence, $\Phi_1(\valpha,\vv)$ is unsatisfiable). 
    \end{itemize}
\end{itemize}

Overall, let $\phi$ be the hard instance from Section~\ref{sec:fixed-order},
which is hard in the variable order $\sigma$. The construction ensures that
$\Psi$ is unsatisfiable and that, for every ordering of the variables $\vu,\vv$,
there is a Boolean assignment to the $\vu$-variables and a relabelling of the
$\vv$-variables by the $\vx$-variables that produces $\phi$ in the order
$\sigma$. This yields a lower bound for all variable orders. Based on this idea,
we construct $\Psi$ in two stages.

\para{Stage 1: local realisation of the clauses.}
 This stage should be viewed as \emph{a lift} of the original instance.
We define the variables $\vu,\vv$ and the first part $\Phi_1$ of the
lifted formula. There are $N$ variables $\vv=\{v_1,\ldots,v_N\}$, which will
serve as renamed copies of the original variables $\vx=\{x_1,\ldots,x_N\}$.
The role of the $\vu$-variables is to specify, clause by clause, how a clause
of the original formula $\phi$ is realised on the $\vv$-variables.

For every clause
\[
C=(\ell_1\vee \ell_2\vee \ell_3)
\]
of $\phi$, and for every triple of distinct indices $i,j,k\in[N]$, we
introduce a selector variable $u_{C,i,j,k}$. Hence, there are $O(M \cd N^3)$-many $\vu$ variables.
Intuitively, setting
$u_{C,i,j,k}=1$ means that, in the lifted copy of the clause $C$, the variables
underlying the literals $\ell_1,\ell_2,\ell_3$ are represented by
$v_i,v_j,v_k$, respectively. We enforce this local copy of $C$ by adding the
clause
\begin{equation}
\neg u_{C,i,j,k}\vee
v_i^{\sign(\ell_1)}\vee
v_j^{\sign(\ell_2)}\vee
v_k^{\sign(\ell_3)} .
\label{eq:13.1}
\end{equation}
Thus, whenever the selector $u_{C,i,j,k}$ is set to $1$, the formula contains
the clause obtained from $C$ by replacing
$\var(\ell_1),\var(\ell_2),\var(\ell_3)$ with $
v_i,v_j,v_k$, respectively, while preserving the signs of the literals.

\begin{definition}
    The formula $\Phi_1$ is the conjunction
    of all clauses of the form~\eqref{eq:13.1}, for all clauses $C$ in $\phi$ and all triples $i\neq j\neq k\in[N]$.
\end{definition}

At this stage, different clauses may choose
their triples independently, so the same original variable might be represented
by different $\vv$-variables in different clauses, or two different original
variables might be represented by the same $\vv$-variable. The role of the
second stage will be to add a consistency formula $\Phi_2$ over the
$\vu$-variables, forcing the selected local realisations to come from one
global relabelling of the original variables by the $\vv$-variables. 

\para{Stage 2: enforcing global consistency.}
The purpose of the second part $\Phi_2$ of the lifted formula is to ensure that the local choices made by the selector variables
$\vu$ in Stage~1 are compatible with a single global relabelling of the
original variables $\vx$ by the new variables $\vv$.

Recall that, for every clause
\[
C=(\ell_1\vee \ell_2\vee \ell_3)
\]
and every triple of distinct indices $i,j,k\in[N]$, the variable
$u_{C,i,j,k}$ says that the variables underlying the literals
$\ell_1,\ell_2,\ell_3$ are realised by $v_i,v_j,v_k$, respectively. Thus, if
$u_{C,i,j,k}=1$, it defines the partial map
\[
\pi_{C,i,j,k}
=
\{\,v_i\mapsto \var(\ell_1),\;
    v_j\mapsto \var(\ell_2),\;
    v_k\mapsto \var(\ell_3)\,\}.
\]
We say that two selector variables $u_{C,i,j,k}$ and $u_{D,a,b,c}$ are
compatible if the union of their corresponding partial maps is still a
one-to-one partial map from $\vv$ to $\vx$. Equivalently, they are incompatible
if either the same $\vv$-variable is mapped to two different $\vx$-variables, or
two different $\vv$-variables are mapped to the same $\vx$-variable.

The formula $\Phi_2$ enforces the following two conditions.
\begin{enumerate}
    \item \label{it:1371} For every clause $C$ of $\varphi$, exactly one of the variables
    $u_{C,i,j,k}$ is set to $1$. We add the clauses
    \[
        \bigvee_{i,j,k\in[N]\text{ distinct}} u_{C,i,j,k}
    \]
    and, for every two distinct triples $(i,j,k)\neq(i',j',k')$,
    \[
        \neg u_{C,i,j,k}\vee \neg u_{C,i',j',k'} .
    \]

    \item \label{it:1372} We rule out incompatible local choices. Namely, for every two clauses
    $C,D$ of $\phi$, and every two triples of distinct indices $(i,j,k)$ and
    $(a,b,c)$, if the two partial maps
    \[
        \pi_{C,i,j,k}
        \qquad\text{and}\qquad
        \pi_{D,a,b,c}
    \]
    are incompatible, then we add the clause
    \[
        \neg u_{C,i,j,k}\vee \neg u_{D,a,b,c}.
    \]
\end{enumerate}

\begin{definition}
    Let $\Phi_2$ be the conjunction of all clauses introduced above in \Cref{it:1371,it:1372}
\end{definition}

Thus, $\Phi_2$ has size polynomial in $N$ and $M$: the number of selector
variables is at most $MN^3$, and the consistency clauses are obtained by
checking pairs of selector variables.

By construction, an assignment to the $\vu$-variables satisfies $\Phi_2$ if
and only if, for every clause $C$, it chooses exactly one local realisation of $C$, and all chosen local realisations are mutually compatible. Equivalently, the selected partial maps combine into one global one-to-one partial map from the relevant $\vv$-variables to the original variables $\vx$. Thus, whenever $\Phi_2$ is satisfied, the clauses selected in $\Phi_1$ form a relabelled copy of the original formula $\phi$. More formally, we have the following. 

\begin{proposition}
    The CNF $\Psi(\vu,\vv) = \Phi_1(\vu,\vv) \wedge \Phi_2(\vu)$ is unsatisfiable.
\end{proposition}

\begin{proof}
Let $\valpha$ be
an arbitrary Boolean assignment to the $\vu$-variables. If $\valpha$ does not
satisfy $\Phi_2$, then clearly $\valpha$ has no extension to the $\vv$-variables
that satisfies $\Psi$. Thus assume that $\valpha$ satisfies $\Phi_2$.

By the definition of $\Phi_2$, for every clause $C$ of $\phi$ there is a
unique triple $(i,j,k)$ such that $\valpha(u_{C,i,j,k})=1$, and all the partial
maps selected in this way are mutually compatible. In other words, these partial maps combine into a single one-to-one partial map from the relevant $\vv$-variables to the original variables $\vx$. Equivalently, there is a single one-to-one partial map $\pi$ from the
$\vv$-variables to the original variables $\vx$ such that, for every clause
$C=(\ell_1\vee \ell_2\vee \ell_3)$ of $\phi$, if $u_{C,i,j,k}$ is the selector that is set to $1$, then
$\pi(v_i)=\var(\ell_1)$, $\pi(v_j)=\var(\ell_2)$, and
$\pi(v_k)=\var(\ell_3)$.
 
Now restrict $\Phi_1$ by the assignment $\valpha$. Let
$C=(\ell_1\vee \ell_2\vee \ell_3)$ be a clause of $\phi$, and let
$(i,j,k)$ be the unique triple selected for $C$. All clauses of $\Phi_1$
corresponding to unselected triples are satisfied, since their selector
variable is set to $0$. The selected triple leaves the clause
\[
v_i^{\sign(\ell_1)}\vee v_j^{\sign(\ell_2)}\vee v_k^{\sign(\ell_3)} .
\]
By the global compatibility of the selected partial maps, these remaining
clauses are exactly a relabelled copy of the clauses of $\phi$, with signs
preserved. Since $\phi$ is unsatisfiable, no assignment to the
$\vv$-variables satisfies $\Phi_1(\valpha,\vv)$. Hence no extension of
$\valpha$ satisfies $\Psi$. Since $\valpha$ was arbitrary, $\Psi$ is
unsatisfiable.
\end{proof}

\subsection{Lower Bound against roABP-$\lIPS$ for All Variable Orders}

We are now ready to prove the main lower bounds result. 
\begin{theorem}[Main lower bound]\label{thm:main-final-lower-bound}
    The CNF $\Psi(\vu,\vv) = \Phi_1(\vu,\vv) \wedge \Phi_2(\vu)$   requires \textup{roABP}-$\lIPS$ refutations of width $2^{N^{\Omega(1)}}$ in any variable ordering.
\end{theorem}

 \begin{proof}
Let $<$ be an arbitrary variable ordering of
the variables of $\Psi$, and suppose that $\Psi$ has an
$\roABP$-$\lIPS$ refutation of width $w$ in this order. Recall that $N$ is the number of $\vv$-variables, which is the number of variables in the original hard instance $\phi=\phi_0\land\phi_1$ denoted $x_1,\dots,x_N$. Consider the order induced by $<$ on the $\vv$-variables, and write it as 
\[
v_{i_1}<v_{i_2}<\cdots<v_{i_N}.
\]
We define a restriction depending only on this induced order. First relabel
$v_{i_r}$ as $x_r$, for every $r\in[N]$. Next assign the $\vu$-variables as
follows. For every clause $C=(\ell_1\vee \ell_2\vee \ell_3)$ of $\phi$ (in 
the $\vx$-variables) 
set to $1$ the unique $u_{C,i,j,k}$ for which $v_i,v_j,v_k$ are the
$\vv$-variables relabelled as $\var(\ell_1),\var(\ell_2),\var(\ell_3)$,
respectively. Set all other variables associated with $C$ to $0$.

This assignment to the $\vu$-variables satisfies $\Phi_2$: it selects 
exactly one triple for each clause, and all selected triples arise from the 
single global relabelling $v_{i_r}\mapsto x_r$. Therefore all compatibility 
clauses in $\Phi_2$ are satisfied. Under the same restriction, every clause of $\Phi_1$ whose selector
$u_{C,i,j,k}$ is assigned $0$ is satisfied by the literal $\neg u_{C,i,j,k}$.
For each original clause $C=(\ell_1\vee\ell_2\vee\ell_3)$ of $\phi$, exactly
one variables $u_{C,i,j,k}$ is assigned $1$, and the corresponding clause \eqref{eq:13.1} of $\Phi_1$ reduces to
\[
v_i^{\sign(\ell_1)}\vee v_j^{\sign(\ell_2)}\vee v_k^{\sign(\ell_3)}.
\]
By the definition of the assignment to the $\vu$-variables and the relabelling
$v_{i_r}\mapsto x_r$, this reduced clause is precisely the original clause
$C$. Hence the restriction of $\Phi_1$ is exactly $\phi$.
Thus the restricted formula is 
precisely $\phi$, with its variables
ordered as $x_1<\cdots<x_N$.

Applying this restriction to the assumed $\roABP$-$\lIPS$ refutation of
$\Psi$ gives an $\roABP$-$\lIPS$ refutation of $\phi$ in the order
$x_1<\cdots<x_N$. The width does not increase under substituting constants
for the $\vu$-variables and relabelling the remaining $\vv$-variables: in the roABP, this only substitutes constants into edge labels and renames 
variables, while preserving the layered computation.

Hence $\phi$ has an $\roABP$-$\lIPS$ refutation of width at most $w$
in the hard order $x_1<\cdots<x_N$. By Theorem~\ref{thm:fixed-order-lbs},
we get $w\geq 2^{N^{\Omega(1)}}$. Since the original order $<$ was arbitrary,
the same lower bound holds for every variable ordering.
\end{proof}
    
%----------------------------------------------------

\section{Simulations and Upper Bounds}
\label{sec:sim-and-upper-bound}
Here we provide information on the strength of roABP-$\lIPS$. We show a simulation of a (dynamic) tree-like  PC refutation systems in which each proof-line is written as a roABP in a fixed global order. Then we provide polynomial-size upper bounds for Tseitin and the functional pigeonhole principle.

%----------------------------------------------------
\subsection{roABP-IPS$_{\mathbf{LIN}}$ Simulates Tree-Like Polynomial Calculus Length}
%----------------------------------------------------

In this section we show that roABP-$\lIPS$ simulates the tree-like Polynomial Calculus as long as the initial polynomial equalities can be written as small-width roABPs. The simulation shows that for simple axioms (computable by constant width roABPs) a length $\ell$ Polynomial Calculus refutation can be transformed into an roABP-$\lIPS$ refutation of width $O(\ell)$. This result holds irrespective of the way the intermediate polynomials are represented in the tree-like Polynomial Calculus refutation as long as they are derived using the two local rules of the system.

\begin{proposition}
    Let $p_1 = 0,\ldots,p_m = 0$ be an unsatisfiable set of polynomial equalities, and suppose each polynomial $p_i$ is computable by width $w$ roABP in some fixed variable ordering.
    
    If there is a tree-like Polynomial Calculus refutation of $p_1 = 0,\ldots, p_m = 0$ of length $l$, then there is an $\mathrm{roABP}$-$\lIPS$ refutation of $p_1 = 0,\ldots, p_m = 0$ of width at most $\ell\cdot w$.
\end{proposition}

\begin{proof}
    Let $q_1,\ldots,q_\ell$ be a tree-like Polynomial Calculus refutation of $p_1 = 0,\ldots,p_m = 0.$ Denote by $s_i$ the size of the sub-tree of the derivation tree rooted at $q_i$. We prove by induction that $q_i$ has an roABP-$\lIPS$ derivation of width at most $s_i\cdot w$ for each $i\in [\ell]$.

    If $q_i = p_j$ for some $j\in [m]$, the claim is clear, since $p_j$ can be computed by width $w$ roABP. If $q_i = xq_j$ for some $j < i$ and some variable $x$, the claim is also clear, since if $q_j$ has an roABP-$\lIPS$ derivation of width $w'$, then $q_i$ has also, since multiplying the derivation by a single variable does not increase its width. Finally, suppose that $q_i = aq_j + bq_k$ for some $j,k <i$ and $a,b\in\F$. By induction assumption, $q_j$ and $q_k$ have roABP-$\lIPS$ derivations of width at most $s_j\cdot w$ and $s_k\cdot w$ respectively. Hence $q_i$ has an roABP-$\lIPS$ derivation of width at most $(s_j + s_k)\cdot w$, which is less than $s_i\cdot w$, since the Polynomial Calculus refutation is tree-like.
\end{proof}

%----------------------------------------------------
\subsection{Upper Bounds for Tseitin Formulas and the Functional Pigeonhole Principle}
%----------------------------------------------------

Lastly, we discuss upper bounds for the proof system roABP-$\lIPS$. We will observe that Tseitin formulas have polynomial width roABP-$\lIPS$ refutations in any variable ordering and show that functional pigeonhole principle has polynomial width roABP-$\lIPS$ refutations in a particular ordering of the involved variables.

In algebraic proof complexity, the restriction to Boolean domain is often imposed by including the Boolean axioms $x^2 - x$ for all variables as is done also by . This forces the variables to take values from the set $\{0,1\}$. This is not however the only possible choice; another common choice is to represent the Boolean values by $\{\pm 1\}$ by including the axioms $x^2 - 1$ for all variables. Although there is a simple linear transformation between these different representation certain complexity measures are sensitive to the choice of representation. For an example the sparsity, or monomial-size, measure can blow-up in this transformation; Tseitin formulas have polynomial sparsity Nullstellensatz refutations in the $\{\pm 1\}$ basis \cite{Gri98}, but over $\{0,1\}$ require exponential sparsity even in the stronger Sum-of-Squares proof system \cite{AH19}. 

The roABP-$\lIPS$ refutations on the other hand are not sensitive for this choice of basis: the linear transformation is applied individually for each variable and thus does not increase the width of the underlying roABPs. This observation shows that the lower bounds presented in this paper in fact give lower bounds over any representation of the Boolean values. Conversely, any upper bounds translate between the different Boolean bases. 

Recall the Tseitin formulas $\Ts(G,\nu)$ \cite{Tse68}. A \emph{charge} on an undirected graph $G$ is mapping $\nu$ that assigns a value $\nu(u)\in \{0,1\}$ to each vertex $u\in V(G)$. The charge is odd, if the number of vertices with charge $1$ is odd. Tseitin's formulas are defined over set of variables $x_e$ for each edge $e\in E(V)$ and consists of the parity formulas
\[
\bigoplus_{\substack{e\in E(V): \\u\in e}} x_e = \nu(u)
\]
for each $u\in V(G)$. The set of formulas is unsatisfiable for any graph with an odd charge. If the graph is of constant degree, the parity formulas can be written as small-width CNFs. Observe, however, that the formulas can always be written with small-width roABPs in any variable ordering. This observation allows us to state the following proposition.

\begin{proposition}
    For any graph $G$ with an odd charge $\nu$ the Tseitin formula $\Ts(G,\nu)$ has a width $\poly(n)$ roABP-$\lIPS$ refutation in any variable ordering.
\end{proposition}

\begin{proof}
    Grigoriev has shown that Tseitin formulas have polynomial sparsity Nullstellensatz refutation over the $\{\pm 1\}$ basis. The sparse polynomials in the refutation can be written as low-width roABPs in any variable ordering, one path for each monomial. After applying the linear transformation $x\mapsto (1 - x)/2$ to each variable, we are left with a low-width roABP-$\lIPS$ refutation of the Tseitin formula in the $\{0,1\}$ basis. 
\end{proof}

Secondly, we will show that the functional pigeonhole principle can be refuted in small width roABP-$\lIPS$ in a particular variable ordering. Recall that the propositional encoding $\FPHP^{n+1}_n$ of the functional pigeonhole principle consists of the following clauses:
\begin{itemize}
    \item $\bigvee_{k\in [n]}x_{ik}$ for all $i\in [n+1]$;\hfill (pigeon axioms)
    \item $\neg x_{ik}\vee \neg x_{jk}$ for all distinct $i,j\in [n+1]$ and all $k\in [n]$;\hfill (pigeonhole axioms)
    \item $\neg x_{ik}\vee \neg x_{i\ell}$ for all $i\in [n+1]$ and all distinct $k,\ell\in [n]$.\hfill (functionality axioms)
\end{itemize}

We consider the variable ordering where variables are ordered by the pigeonhole: $x_{ik} < x_{j\ell}$, whenever $k < \ell$. The argument given below is similar to the upper bound given for constant-depth IPS in \cite{HLT24}, which in turn borrowed from the upper bounds in \cite{GH03} and \cite{RT07}.

\begin{proposition}
    $\FPHP^{n+1}_n$ has width $\poly(n)$ roABP-$\lIPS$ refutations in the variable ordering by the pigeonhole.
\end{proposition}

\begin{proof}
    First note that the pigeon axioms can be written equivalently in the form $1 - \sum_{k\in [n]}x_{ik}$
    and that this equivalence is easily derivable from the standard translation of the pigeon axioms and the functionality axioms as
    \[
    1 - \sum_{k\in [n]}x_{ik} = \prod_{k\in [n]}(1 - x_{ik}) - \sum_{k < \ell}x_{ik}x_{i\ell}\prod_{\ell' > \ell}(1 - x_{i\ell'}).
    \]
    The weights on the functionality axioms above have a constant width read-once oblivious algebraic branching programs. The sum of all these representations of the pigeon axioms equals
    \[
    n+1 - \sum_{i\in [n+1]}\sum_{k\in [n]}x_{ik}
    \]
    Define the following short-hand: denote by $y_k$ the polynomial $\sum_{i\in [n+1]} x_{ik}$. This polynomial is Boolean valued as $y_k^2 - y_k$ can be easily derivable from the hole axioms and the Boolean axioms. Now write the sum of all pigeon axioms in the form
    \[
    n+1 - \sum_{k\in [n]}y_k.
    \]
    This is an unsatisfiable subset sum instance in the $\vy$ variables; its unique multilinear refutation is symmetric and can thus be written as a linear combination of elementary symmetric polynomials (\cite{FSTW21} gives an explicit expression). As elementary symmetric polynomials have small width read-once oblivious algebraic branching programs, this subset sum instance has small width roABP-$\lIPS$ refutations in the $\vy$ variables in the variable ordering $y_1 < \ldots < y_n$.

    Finally apply the substitution $y_k \mapsto \sum_{i\in [n+1]} x_{ik}$ to the refutation of the subset sum instance to obtain a refutation of the functional pigeonhole principle. The width of the refutation stays small since the chosen variable ordering interleaves with the ordering on the $\vy$ variables: $x_{ik} < x_{j\ell}$ for any $k < \ell$.
\end{proof}

\small

\bibliographystyle{alpha}
\bibliography{reference}

\end{document}